\documentclass[letterpaper,11pt]{article}
\usepackage[utf8]{inputenc}
\usepackage{amsmath, amsthm, amssymb, thm-restate}
\usepackage[table,xcdraw]{xcolor}
\usepackage[vlined,linesnumbered]{algorithm2e}
\usepackage{setspace}
\usepackage{mathtools}
\usepackage{algpseudocode}
\usepackage[numbers]{natbib}
\usepackage{comment} 
\usepackage{tcolorbox} 
\usepackage{xfrac}
\usepackage{hyperref}
\usepackage{multirow}
\usepackage{caption}
\usepackage{bm}
\usepackage{newfloat}
\usepackage{enumitem}
\usepackage{dblfloatfix} 
\usepackage{wrapfig}
\usepackage{dsfont}
\usepackage{mdframed}
\usepackage{multirow}
\usepackage{tikz}

\usepackage{fancyhdr}
\usepackage[noabbrev,nameinlink]{cleveref}

\usepackage[margin=1in]{geometry}

\allowdisplaybreaks

\definecolor{mygreen}{RGB}{10,110,230}
\definecolor{myred}{RGB}{10,110,230}

\hypersetup{
     colorlinks=true,
     citecolor= mygreen,
     linkcolor= myred
}

\newcommand{\ev}{\mathbb{E}} 
\newcommand{\g}{g} 
\newcommand{\pr}{\text{Pr}} 
\newcommand{\price}{f} 
\newcommand{\vecx}{\mathbf{x}} 
\newcommand{\vecy}{\mathbf{y}} 
\newcommand{\numpieces}{k} 
\newcommand{\classbackup}{C_b} 
\newcommand{\classnobackup}{C_s} 
\newcommand{\classnomatch}{C_\bot} 
\newcommand{\hclassbackup}{h_b} 
\newcommand{\hclassnobackup}{h_s} 
\newcommand{\hclassnomatch}{h_\bot} 
\newcommand{\Hclassbackup}{H_b} 
\newcommand{\Hclassnobackup}{H_s} 
\newcommand{\Hclassnomatch}{H_\bot} 
\newcommand{\x}{x} 
\newcommand{\apx}{\ensuremath{0.5469}} 
\newcommand{\y}{y} 
\newcommand{\B}{\text{B}} 
\newcommand{\I}{\text{I}} 
\newcommand{\ranking}{{\textsc{Ranking}}} 
\newcommand{\M}{\textsc{R}}

\newcommand{\opt}{M^*}

\newcommand{\E}[0]{\ensuremath{\mathbb{E}}}

\SetKwRepeat{Do}{do}{while}

\renewcommand{\epsilon}{\ensuremath{\varepsilon}}

\let\originalleft\left
\let\originalright\right
\renewcommand{\left}{\mathopen{}\mathclose\bgroup\originalleft}
\renewcommand{\right}{\aftergroup\egroup\originalright}

\crefname{lemma}{Lemma}{Lemmas}
\crefname{theorem}{Theorem}{Theorems}
\crefname{property}{Property}{Properties}
\crefname{claim}{Claim}{Claims}
\crefname{result}{Result}{Results}
\crefname{definition}{Definition}{Definitions}
\crefname{observation}{Observation}{Observations}
\crefname{proposition}{Proposition}{Propositions}
\crefname{assumption}{Assumption}{Assumptions}
\crefname{line}{Line}{Lines}
\crefname{figure}{Figure}{Figures}
\creflabelformat{property}{(#1)#2#3}
\crefname{equation}{}{}
\crefname{section}{Section}{Sections}
\crefname{appendix}{Appendix}{Appendices}
\crefname{algCounter}{Algorithm}{Algorithms}
\Crefname{algCounter}{Algorithm}{Algorithms}

\newtheorem{lemma}{Lemma}[section]
\newtheorem{theorem}[lemma]{Theorem}

\newtheorem{corollary}[lemma]{Corollary}

\newtheorem{definition}[lemma]{Definition}
\newtheorem{claim}[lemma]{Claim}
\newtheorem{fact}[lemma]{Fact}

\newtheorem*{remark*}{Remark}

\newtheorem*{definition*}{Definition}

\makeatletter
\def\thm@space@setup{\thm@preskip= 0.2cm \thm@postskip=\thm@preskip}
\makeatother

\definecolor{mygreen}{RGB}{20,155,20}
\definecolor{myred}{RGB}{195,20,20}
\definecolor{linkcolor}{RGB}{0,0,230}
\definecolor{mylightgray}{RGB}{230,230,230}
\definecolor{verylightgray}{RGB}{240,240,240}
\definecolor{commentcolor}{RGB}{120,120,120}

\SetCommentSty{mycommfont}

\hypersetup{colorlinks=true, citecolor= mygreen, linkcolor= myred, urlcolor= mygreen}

\newcounter{myalgctr}

\newenvironment{tbox}{\par\addvspace{0.2cm}\begin{tcolorbox}[width=\textwidth, boxsep=2pt, left=1pt, right=1pt, top=4pt, boxrule=1pt, arc=0pt, colback=white, colframe=black]}{\end{tcolorbox}}

\newenvironment{tboxh}{\par\addvspace{0.2cm}\begin{tcolorbox}[width=\textwidth, boxsep=2pt, left=1pt, right=1pt, top=4pt, boxrule=1pt, arc=0pt, colback=white, colframe=black, float=t]}{\end{tcolorbox}}

\newenvironment{graytbox}{
\par\addvspace{0.1cm}
\begin{tcolorbox}[width=\textwidth,
                  boxsep=5pt,
                  left=1pt,
                  right=1pt,
                  top=2pt,
                  bottom=2pt,
                  boxrule=0pt,
                  arc=0pt,
                  colback=mylightgray,
                  colframe=black,
                  ]
}{
\end{tcolorbox}
}

\newcommand{\tboxhrule}{\vspace{0.1cm} \hrule \vspace{0.2cm}}

\newenvironment{titledtbox}[1]{\begin{tbox}#1 \tboxhrule}{\end{tbox}}
\newenvironment{titledtboxh}[1]{\begin{tboxh}#1 \tboxhrule}{\end{tboxh}}

\newcounter{protocolcounter}
\newenvironment{protocol}{\refstepcounter{protocolcounter}\begin{tbox}\textbf{Algorithm \theprotocolcounter:}}{\end{tbox}}
\crefname{protocolcounter}{Algorithm}{Algorithms}

\makeatletter
\renewcommand{\paragraph}{%
  \@startsection{paragraph}{4}%
  {\z@}{10pt}{-1em}%
  {\normalfont\normalsize\bfseries}%
}
\makeatother

\makeatletter
\patchcmd{\@algocf@start}
  {-1.5em}
  {0pt}
  {}{}
\makeatother

\title{Improved Approximation for Ranking on General Graphs}

 \author{
 Mahsa Derakhshan \\ {\em Northeastern University} \and 
 Mohammad Roghani  \\ {\em Stanford University} \and
 Mohammad Saneian  \\ {\em Northeastern University} \and
 Tao Yu  \\ {\em Northeastern University}
 }
\date{}

\begin{document}
\maketitle

\setcounter{page}{1}

In this paper, we study \textsc{Ranking}, a well-known randomized greedy matching algorithm, for general graphs. The algorithm was originally introduced by Karp, Vazirani, and Vazirani [STOC 1990] for the online bipartite matching problem with one-sided vertex arrivals, where it achieves a tight approximation ratio of $1 - 1/e$. It was later extended to bipartite graphs with random vertex arrivals by Mahdian and Yan [STOC 2011] and to general graphs by Goel and Tripathi [FOCS 2012].  The \textsc{Ranking} algorithm for general graphs is as follows: a permutation $\sigma$ over the vertices is chosen uniformly at random. The vertices are then processed sequentially according to this order, with each vertex being matched to the first available neighbor (if any) according to the same permutation $\sigma$.

While the algorithm is quite well-understood for bipartite graphs—with the approximation ratio lying between 0.696 and 0.727—its approximation ratio for general graphs remains less well-characterized despite extensive efforts. Prior to this work, the best-known lower bound for general graphs was $0.526$ by Chan et al. [TALG 2018], improving on the approximation ratio of $0.523$ by Chan et al. [SICOMP 2018]. The upper bound, however, remains the same as that for bipartite graphs.

In this work, we improve the approximation ratio of \textsc{Ranking} for general graphs to $0.5469$, up from $0.526$. This also surpasses the best-known approximation ratio of $0.531$ by Tang et al. [JACM 2023] for the oblivious matching problem. Our approach builds on the standard primal-dual analysis used for online matchig. The novelty of our work lies in proving new structural properties of \textsc{Ranking} on general graphs by introducing the notion of the \emph{backup} for vertices matched by the algorithm. For a fixed permutation, a vertex's backup is its potential match if its current match is removed. This concept helps characterize the rank distribution of the match of each vertex, enabling us to eliminate certain \emph{bad events} that constrained previous work.

\newpage



\section{Introduction}

In this paper, we study \textsc{Ranking}, a well-known randomized greedy matching algorithm, for general graphs. The algorithm was originally introduced by Karp, Vazirani, and Vazirani \cite{karp1990optimal} for the online bipartite matching problem with one-sided vertex arrivals, where it achieves a tight approximation ratio of $1 - 1/e$~\cite{karp1990optimal, GoelM08, BirnbaumM08, DBLP:conf/soda/DevanurJK13}. 
The algorithm was later extended to bipartite graphs with random vertex arrivals by Mahdian and Yan~\cite{mahdian2011online} and to general graphs by Goel and Tripathi~\cite{goel2012matching}. 

The \ranking{} algorithm falls within a broader class of vertex-iterative randomized greedy matching algorithms. A vertex-iterative randomized greedy matching algorithm begins by generating a uniformly random permutation $\pi$ of the vertices. The vertices are then processed sequentially according to this order, with each vertex being matched to an available neighbor based on a predefined selection rule. In \ranking{} for general graphs, this rule ensures that a vertex is matched to the first unmatched neighbor that appears earlier in the same permutation $\pi$. This algorithm can be implemented in two well-studied models:

\vspace{-1 mm}
\paragraph{Oblivious Matching.}  This model also known as the query commit model was first introduced by~\cite{goel2012matching}. Consider an $n$ vertex graph whose edges are unknown. The existence of edges will only be revealed via queries. The algorithm picks a permutation of unordered pairs of vertices.  It then constructs a greedy matching by probing the edges one-by-one according to the
permutation. Particularly, when a pair of vertices is probed, if an edge exists between them and both are unmatched, the edge joins the matching. Currently, the best-known bounds for general graphs fall between 0.531 by Tang et al.~\cite{towardsabetterunderstanding} and 0.75~\cite{goel2012matching}. 
The 0.531 approximation is achieved via a vertex iterative algorithm. Note that while any vertex-iterative algorithm can be implemented in this model, the best algorithm in this model may not be vertex iterative.

\vspace{-1 mm}
\paragraph{Online Matching with Random Vertex Arrivals.} In this model, the vertices of an unknown graph arrive online in a uniformly random order. Upon the arrival of a vertex, its edges to the previous vertices are revealed. At this time, the algorithm needs to decide irrevocably whether to include one of these edges (if the other end-point is still available) in the matching or lose them forever. This model was first introduced by \cite{mahdian2011online} for bipartite graphs and later extended to general graphs by Chan et al.~\cite{chan2014ranking}. It can also be viewed as the random arrival version of the fully online model introduced by Huang et al.,~\cite{Howtomatchwhenallverticesarriveonline}. The best-known approximation ratio for general graphs in this setting is $0.526$ by Chan et al.~\cite{chan2018analyzing}, who improve the $0.523$ approximation ratio of~\cite{chan2014ranking}.

The \ranking{} algorithm studied in this paper is quite well understood for bipartite graphs, with its approximation ratio known to lie between 0.696~\cite{mahdian2011online} and 0.727~\cite{conf/stoc/KarandeMT11}. However, for general graphs, its approximation ratio remains less well-characterized despite extensive research efforts~\cite{DyerF91, aronson1995randomized, chan2014ranking, goel2012matching, poloczek2012randomized, towardsabetterunderstanding, 0.505_simplified_Ranking}. Currently, the best-known bounds for general graphs fall between $0.526$~\cite{chan2014ranking} and $0.727$ \cite{conf/stoc/KarandeMT11}.

\vspace{-1 mm}
\paragraph{Our Contribution.} In this work, we make progress by improving the approximation ratio of \ranking{} for general graphs. We present a refined analysis of \ranking{} and establish that it achieves an approximation ratio of at least 0.5469. Our work improves the best-known approximation ratio for online matching with random vertex arrivals and the oblivious matching problem respectively from $0.526$  by~\cite{chan2014ranking} and 0.531~\cite{towardsabetterunderstanding}. More formally, we prove the following theorem:

\vspace{0.5em}
\begin{graytbox}
\begin{theorem}
    The approximation ratio of \ranking{} on general graphs is at least 0.5469.
\end{theorem}
\end{graytbox}
\vspace{0.5em}

\section{Our Techniques}

We build on the standard primal-dual framework commonly used to analyze the \ranking{} algorithm in both general and bipartite graphs (see \cite{mahdian2011online, chan2018analyzing, chan2014ranking, DBLP:conf/soda/0002PTTWZ19, DBLP:conf/soda/DevanurJK13, DBLP:journals/talg/HuangTWZ19, peng2025revisiting, xu2025bounding}, and references therein). For each edge $(u, v)$ in the output of \ranking{}, we divide its gain (which is 1) between its endpoints $u$ and $v$ according to a function we define later. Let $g(v)$ denote the gain assigned to vertex $v$. We show that for every edge $(u, u^*)$ in a maximum matching of $G$, the expected total gain of its endpoints, $\E[g(u) + g(u^*)]$, is at least $\apx$. This establishes an approximation ratio of $\apx$ for the algorithm.

Our work is particularly inspired by the economics-based analysis by Eden et al.~\cite{eden2020economicbasedanalysisrankingonline} for bipartite graphs. In their work, vertices on one side are referred to as \emph{items} and those on the other as \emph{buyers}.\footnote{This paper presents an alternative, simplified analysis of the classic bipartite \ranking{} algorithm, where only one side of the graph is randomly permuted.} A random permutation $\sigma$ of the items is generated by having each item $v$ independently draw a \emph{rank} $x_v$ uniformly at random from the interval $(0,1)$, and ordering the items by increasing rank. Buyers arrive in an arbitrary order and are matched to their first unmatched neighbor (if any) according to $\sigma$. If an item $v$ is matched to a buyer $u$, the item's gain—interpreted as its \emph{price}—is given by a function $f(x_v)$, and the buyer receives the remaining gain $1 - f(x_v)$. By choosing an appropriate function $f$, they show that for every edge $(u, u^*)$ in the optimal solution, the total expected gain of its endpoints is at least $1 - 1/e$, thus proving a $1 - 1/e$ approximation ratio.

\paragraph{Random Partitioning of Vertices.} One of the challenges in extending this analysis to general graphs is the absence of a natural partition of vertices into items and buyers.  We cannot fix the partition before running the algorithm, as \ranking{} may match two vertices of the same type. However, we observe that since gain sharing is purely for the sake of analysis, such a partition is only needed after the \ranking{} algorithm has been completed. Importantly, the only edges relevant to the gain-sharing analysis are those used by \ranking{} and those in the optimal solution. Since the union of two matchings cannot form an odd-length cycle, the relevant edges form a bipartite graph. This allows us to partition the vertices into two sets—buyers and items—such that all relevant edges go between the two groups. It is also crucial that, for any vertex $v$, the outcome of the \ranking{} is independent of whether any vertex $v$ is labeled as an item or a buyer. Therefore, after partitioning the vertices into two parts, we randomly label one as the set of items (with probability 0.5) and the other as the set of buyers, preserving the symmetry needed for the analysis.
When analyzing the total expected gain of $u$ and $u^*$, the symmetrical partitioning allows us to assume without loss of generality that $u$ is the buyer and $u^*$ the item. Simillar to the economic-based approach of~\cite{eden2020economicbasedanalysisrankingonline} our gain-sharing approach involves assigning a price in range $(0,1)$ to each item $v$, which is an increasing function of its rank. If $v$ matches $u$, then $v$'s gain equals its price and the reminder goes to $u$. 

\paragraph{Comparing the Outcome of \ranking{} With and Without $u^*$.}  
To lower-bound the total expected gain of $u$ and $u^*$, $\E[g(u) + g(u^*)]$, we begin by removing $u^*$ (assumed to be an item) from the graph and drawing a uniformly random permutation over the remaining vertices, denoted by $\sigma_{-u^*}$. This permutation is obtained by having each vertex $v \in V \setminus \{u^*\}$ independently draw a \emph{rank} $x_v$ uniformly at random from the interval $(0,1)$, and ordering the vertices by increasing rank. We then run the \ranking{} algorithm on $\sigma_{-u^*}$. 
Next, we reintroduce $u^*$ into the graph by independently drawing its rank $x_{u^*}$ (independent of all other ranks) and inserting it into the permutation, resulting in the full permutation $\sigma$.  Our analysis involves comparing the matching produced by \ranking{} on $\sigma_{-u^*}$ and $\sigma$, denoted respectively by $\M(\sigma_{-u^*})$ and $\M(\sigma)$.

A crucial property used in the analysis of \ranking{} on bipartite graphs is the following: if $u$ is matched to a vertex $v$ in $\M(\sigma_{-u^*})$, then in $\M(\sigma)$, it remains matched—possibly to a vertex $v'$ whose rank is no larger than that of $v$. In other words, adding a vertex to one side of a bipartite graph cannot unmatch vertices on the other side. This property is fundamental to the analysis of~\cite{eden2020economicbasedanalysisrankingonline} and is what enables the proof of the $1 - 1/e$ approximation ratio in the bipartite setting. Unfortunately, this property does not extend to general graphs, where adding  $u^*$ can disrupt the structure of the matching and potentially unmatch $u$. The lack of this property is the main reason behind the difficulty of the analysis of \ranking{} in general graphs. Previous works~\cite{towardsabetterunderstanding, Howtomatchwhenallverticesarriveonline}  address this issue through a more involved gain-sharing approach. Rather than simply sharing the gain between the end-points of an edge in the matching, they extend some of the gain to a third vertex which they call the \emph{victim} of this event. In this work, we take an alternative approach. While in our analysis we do not directly deal with this bad event (i.e., $u^*$ unmatching $u$), for the sake of conveying our ideas, we will explain how they can be used to do so. 
Roughly speaking, we show that the probability of the bad event  cannot be too high.
Motivated by this, we introduce the concept of a \emph{backup} for the vertex $u$, which we later formally state in~\cref{defbackup}.

\begin{definition*}[Backup of $u$]
Assume that $u$ has a match $v$ in $\M(\sigma_{-u^*})$, the output of \ranking{} on permutation $\sigma_{-u^*}$. The \emph{backup} of $u$ with respect to $\sigma_{-u^*}$ is the vertex $u$ matches to if $v$ is removed from the permutation. If $u$ becomes unmatched as a result of removing $v$, we say that $u$ has no backup. We use $b$ to denote this backup vertex.
\end{definition*}

We prove a key property regarding the backup of $u$. Specifically, if $u$ has a backup $b$, then adding $u^*$ to the permutation cannot cause $u$ to become unmatched. Furthermore, in $\M(\sigma)$, the rank of the vertex matched to $u$ is at most as large as the rank of $b$. (We formally prove this in \cref{lem:backup}) This implies that the bad event of  $u^*$ unmatching $u$ can only happen when $u$ does not have a backup\footnote{See \Cref{sec:appendix-a} for a detailed comparison of backup vertices with similar definitions in previous works, such as \cite{chan2018analyzing}.}.

Next we argue that conditioned on $u$ not having a backup the distribution of the rank of its match in $\sigma_{-u^*}$ has a property that limits the probability of the bad event happening (i.e. adding $u^*$ unmatches $u$). Note that $u^*$ can only unmatch $u$ if it joins the matching itself. We argue that with a sufficiently large probability $u$ and $u^*$ will both be matched in $\M(\sigma)$ when $u$ is matched in $\M(\sigma_{-u^*})$ and has no backup. Basically, our goal is to show that the following probability is sufficiently large where the probability is taken over the randomization of the rank vector $\mathbf{x}$ which resulted in permutation $\sigma$:
$$\Pr_{\mathbf{x}}[\{u, u^*\} \subset \M(\sigma)  \mid u\in \M(\sigma_{-u^*}) \text{ and } u \text{ has no backup}].$$
This is due to two crucial observations:
\begin{enumerate}
    \item Suppose $u$ is matched to a vertex $v$ in $\M(\sigma_{-u^*})$ and has no backup with respect to $\sigma_{-u^*}$.  In this scenario, increasing the rank of $v$ while  keeping the relative order of the rest of the vertices unchanged, does not change the matching $\M(\sigma_{-u^*})$. We use this observation to infer that conditioned on the event that $u$  is matched with no backup the rank of its match is monotonically increasing.  The randomization here comes from the rank vector $\mathbf{x}$.  (This is formally stated as~\cref{descretecoro-mootone-nobackup}.)
    \item Suppose $u$ is matched to a vertex $v$ in $\M(\sigma_{-u^*})$ s.t. $x_u < x_v$ meaning that $u$ comes earlier in the permutation. If $x_{u^*}$ falls in range $(x_u, x_v)$, then both $u$ and $u^*$ will be matched in $\M(\sigma)$. Since $x_{u^*}$ is drawn independently of $x_v$ and $x_u$, this happens with probability $x_v-x_u$. (This is formally stated in~\cref{lem:struct-property})
\end{enumerate}

Putting the two observations together and doing some calculations we conclude that given a fixed rank $x_u$ for vertex $u$ and conditioned on $u$ having no backup, with probability at least $(1-x_u)^2/2$ both $u$ and $u^*$ are matched in $\M(\sigma)$. As long as $x_u$ is sufficiently smaller than one, this probability is a constant. For instance if $x_u \leq 0.5$ then this probability is at least $1/8$. 

If one only hopes to show that \ranking{} has an approximation ratio of $0.5+\epsilon$ for some constant $0<\epsilon$ following the approach of bounding the probability of the \emph{bad event} allows for a simple analysis. However, our aim is to show a more significant approximation ratio of $\apx$. Therefore, we continue with our gain-sharing analysis. 

\paragraph{Optimizing the Gain-sharing Function.}
When a buyer $u$ is matched to an item $v$, we divide the gain using a function $f(x_u, x_v)$, interpreted as the price of $v$ for buyer $u$. We then define the gains as $g(v) = f(x_u, x_v)$ and $g(u) = 1 - f(x_u, x_v)$. Previous works involving the gain-sharing technique mostly defined function $f(x_u, x_v)$ to depend on only one of the variables $x_u$ or $x_v$, as in \cite{RandomizedPrimalDual, Howtomatchwhenallverticesarriveonline, Thighcompetitiveratioofclassicmatching, towardsabetterunderstanding}, or imposed a symmetry constraint $f(x_u, x_v) = 1 - f(x_v, x_u)$, as in \cite{Onlinevertexweightedbipartitematching}. With the help of our random vertex partitioning technique, we are able to relax $f$ to depend on both $x_u, x_v$ and without imposing the symmetry constraint. This relaxed function $f$ is one of the key ingredients in our $0.5469$-approximation.

To optimize $f$, we formulate a factor-revealing linear program (LP), requiring $f$ to be monotonically increasing in $x_v$ and decreasing in $x_u$. We discretize the rank interval into $k$ equal-sized buckets, giving us $k^2$ variables for $f$.
To help us express the expected gain as a linear function of $f$, we introduce the notion of \emph{profile} for $u$:

\begin{definition*}[Profile of $u$]
The profile of $u$ with respect to permutation $\sigma_{-u^*}$ consists of three pieces of information regarding $u$. It is a tuple $(x_u, x_v, x_b)$, denoting the ranks of $u$, its match $v$ (if any), and its backup (if any). See~\cref{def:profile}. 
\end{definition*}

Fixing $x_u$, we derive a lower bound on $\E[g(u) + g(u^*) \mid x_u]$ as a linear function of $f$, using $O(k^2)$ representative distributions over $u$’s profile, falling into three types:
\begin{itemize}
    \item $u$ is matched but has no backup. The rank $x_v$ of the match is uniformly distributed in $(i/k, 1)$ for some $0 \le i < k$.
    \item  $u$ is matched and has a backup. The match's rank $x_v$ is uniformly distributed in $(i/k, j/k)$, and the backup has rank $x_b = j/k$, for $0 \le i < j \leq k$.  
\item $u$ is unmatched and has no backup (a unique profile).
\end{itemize}

Using the structural properties, we prove we are able to formulate a linear lower bound for the expected gain $\E[g(u) + g(u^*) \mid x_u]$ given each distribution. The final bound is the minimum over these.

This leads to our factor-revealing LP, fully described in~\cref{LP}, and establishes the approximation ratio of $\apx$ for \ranking{}. This ratio is obtained by solving the LP with $k = 100$. \cref{figure:example} illustrates an example for $k = 3$. Notably, the function $f$ is asymmetric—a crucial feature for achieving our approximation guarantee. This need for asymmetry underlies the necessity of partitioning vertices into items and buyers.

\begin{table}[h]
    \centering 
    \renewcommand{\arraystretch}{1.2} 
    \setlength{\tabcolsep}{6pt} 
    \small
    \begin{tabular}{|c||c|c|c|c|}
        \hline
        $i \backslash j$ & [0, 1/3] & [1/3, 2/3] & [2/3, 1]\\ \hline\hline
        [0, 1/3] & 0.469 & 0.563 & 0.563 \\ \hline
        [1/3, 2/3] & 0.469 & 0.500 & 0.500 \\ \hline
        [2/3, 1] & 0.469 & 0.500 & 0.500 \\ \hline
    \end{tabular}
    \caption{Values of $f(i,j)$ for $k=3$  resulting in a lower bound of $0.503$ for the approximation ratio.}
    \label{figure:example}
\end{table}

\section{Preliminaries}

\subsection{Graph and Permutation Notation}
We consider undirected graphs $G = (V, E)$ with $|V| = 2n$ vertices. A \emph{matching} $M \subseteq E$ is a set of edges such that no two edges share an endpoint. A vertex incident to an edge in $M$ is said to be \emph{matched} or \emph{unavailable}; otherwise, it is \emph{unmatched} or \emph{available}. We let $\opt$ denote a fixed maximum matching of $G$. We also assume that $G$ admits a perfect matching of size $n$. We can make this assumption because of the following commonly used fact:
\begin{lemma}[\cite{aronson1995randomized}]\label{wecanassume2n}
    Let $\alpha \in [0,1]$ be a number that lower-bounds the approximation ratio for \ranking{} on all graphs $G^*$ that admit a perfect matching $M^*$. Then, $\alpha$ also lower-bounds the approximation ratio for \ranking{} on any graph $G$.
\end{lemma}

We use random rank vectors $\vecx$ and random permutations $\sigma$ as sources of randomness\footnote{We will eventually use what we call bucketed random permutation vectors as the source of randomness. This is merely a discretization of the uniform real rank vector $x_v$ so that we can build the LP with proper discretization. The relevant notions are presented in Section~5 and are not needed before that.}. Let $\vecx:V\to(0,1)$ be a random vector such that the values $x_v$ for each index are i.i.d. random variables drawn uniformly from $(0,1)$, i.e. $\x_v\overset{iid}{\sim}\mathcal{U}(0,1)$. Each realization of $\vecx$ induces a unique permutation $\sigma_{\vecx}:V\to\{1,...,|V|\}$ such that $\sigma_{\vecx}(u)\leq \sigma_{\vecx}(v)\Leftrightarrow \x_u\leq \x_v$\footnote{There are multiple possible permutations when there exists some $u\neq v$ such that $x_u=x_v$. The probability measure of such an event happening is $0$, so we assume that it does not happen.}. The distribution of $\vecx$ induces a distribution on permutations $\sigma$, which is a uniform distribution on permutations. Equivalently, we can draw a uniform random permutation $\sigma: V \to \{1, \dots, |V|\}$ directly without vector $\vecx$. 

We interpret $\sigma$ as both the arrival order and the preference order of vertices. For a vertex $v$, we write $\x_v$ or $\sigma(v)$ for its rank; which one is used will be clear from the context. For a set $S \subseteq V$, we write $\sigma_{-S}$ to denote the permutation $\sigma$ with the vertices in $S$ removed. That is, $\sigma_{-S}$ is the induced permutation of $\sigma$ with its domain restricted to $V \backslash S$. In particular, we write $\sigma_{-u^*}$ to denote the permutation induced on $V \backslash \{u^*\}$ and $\sigma_{-vu^*}$ to denote the permutation induced on $V \backslash \{v, u^*\}$. Conversely, $\sigma$ can be viewed as a permutation generated by adding $S$ to a permutation $\sigma_{-S}$ by assigning each $v \in S$ the rank $\sigma(v)$ while preserving the relative order of vertices in $V \backslash S$. Similarly, we also use $\vecx_{-S}$ to represent the induced rank vector on $V \backslash S$.

Let $A, B$ be two sets. We will use $A \oplus B$ to denote the symmetric difference of $A$ and $B$, that is, $A \oplus B = (A \backslash B) \cup (B \backslash A)$. 

We assume the standard lexicographic ordering on pairs of numbers $(a, b)$, i.e., for $a, b, c, d \in \mathbb{R}$, $(a, b) \leq (a', b')$ iff $(a < a')$ or $(a = a' \;\wedge\; b \leq b')$.

\subsection{The \ranking{} Algorithm}

In this paper, we use two equivalent views of the \ranking{} algorithm: the vertex-iterative view and the greedy probing (edge-iterative) view.

\paragraph{Vertex-Iterative View.}
The algorithm begins by selecting for each vertex $v$ a rank $\x_v$ in $(0,1)$ uniformly and independently at random. The vector $\vecx$ containing all the $\x_v$ values is then translated into a permutation $\sigma:V\to\{1,\dots,|V|\}$ based on the increasing order of the $x$ values. Then, it processes the vertices in order of increasing $\sigma$. When a vertex $u$ is processed, it is matched to the first unmatched neighbor in its neighborhood $N(u)$, based on the same ordering $\sigma$. If no unmatched neighbor exists, $u$ remains unmatched.

\begin{figure}[h]
\begin{protocol}\label{alg:vertex-ranking-general}
    Vertex-iterative version of the \textsc{Ranking} algorithm for general graphs

    \smallskip\smallskip
    \textbf{Input:} Graph $G = (V, E)$

    \smallskip\smallskip
    \begin{enumerate}[leftmargin=15pt]
        \item Sample a uniform random value $\x_v\in(0,1)$ independently for each $v\in V$
        \item Define permutation $\sigma$ as the unique permutation s.t. $\sigma(u)<\sigma(v)\Leftrightarrow x_u<\x_v$.
        \item Initialize matching $M \gets \emptyset$
        \item Initialize all vertices as unmatched
        \item For each vertex $v \in V$ in order $\sigma$:
        \begin{enumerate}
            \item Enumerate $u$ such that $e = (u, v) \in E$ in order $\sigma$
            \begin{enumerate}
                \item If both $u$ and $v$ are available
                \begin{enumerate}
                    \item Add $(u,v)$ to $M$
                    \item Mark $u$ and $v$ as matched
                    \item break
                \end{enumerate}
            \end{enumerate}
        \end{enumerate}
        \item Return matching $M$
    \end{enumerate}
\end{protocol}
\end{figure}

The algorithm, as defined, is not applicable to the uniform random vertex arrival model\footnote{Here we assume that in the random vertex arrival model, each vertex arrives uniformly at random and may only match with a neighbor that has already arrived.}, as it will always match a vertex with a higher-ranked neighbor. However, one can prove that imposing the restriction that each vertex may only be matched to previously arrived vertices does not change the output graph. Under this restriction, the algorithm becomes applicable to the uniform random vertex arrival model. See \cref{vertexarrivalmodels} for a detailed discussion.
\label{greeedyprobing}
\paragraph{Greedy Probing View}
This is also known as the query commit view. Draw $\sigma: V \to \{1, \dots, |V|\}$ as a uniform random permutation on vertices, either directly or through $\vecx$. $\sigma$ naturally induces an ordering $\sigma_E$ on edges $(u, v) \in E$ by the lexicographic ordering of $(\sigma(u), \;\sigma(v))=\sigma_E(u,v)$.  

\ranking{} probes edges $(u, v)$ in ascending order of $\sigma_E$. Whenever an edge $(u, v)$ is probed and both ends are available, \ranking{} includes $(u, v)$ in the match set and marks both ends as matched. 

Note that while an edge may be probed twice, once at time $\sigma_E(u,v)$ and once at $\sigma_E(v,u)$, the second time never results in a matching. In this paper, whenever we refer to ``the time an edge $(u, v)$ is probed", it makes sense to refer to the first query time. This time can be defined unambiguously as $\min\{\sigma_E(u,v),\;\sigma_E(v,u)\}$.

We here present a simple fact about the greedy proving view: if two different edges $(u,v)$ and $(u,v')$ share a common node, the order in which they get queried is determined solely by the other node which they do not share.
\begin{fact}\label{lexicographicorder}
    Assuming the greedy probing model, let $t$ be the time edge $(u,v)$ is probed, and let $t'$ be the time edge $(u,v')$ gets probed. It follows that $t \leq t'$ if and only if $\sigma(v) \leq \sigma(v')$.
\end{fact}
\begin{proof}
    Let $\leq_{lex}$ denote the lexicographic ordering on pairs. We have the simple fact that for any reals $a, b , c \in \mathbb{R}$
    $$\min\{(a, b), (b, a)\} \leq_{lex} \min\{(c, a), (a, c)\} \iff b \leq c.$$
    As the greedy probing model probes edges with respect to the increasing lexicographic order on tuples, we have 
$$
t_1 = \min\{\sigma_E(u,v),\sigma_E(v,u)\} \leq_{lex} \min\{\sigma_E(u,v'),\sigma_E(v',u)\} = t_2 \Longleftrightarrow \sigma(v) \leq \sigma(v'),
$$
\end{proof}
\subsection{Matching Structures and Terminology}
Let $\sigma$ be a permutation of vertices in $V$ and let $\sigma_{-u^*}$ be a permutation on $V\backslash\{u^*\}$. We use $\M(\sigma)$ to represent the result of running \ranking{} on $G$ with respect to ordering $\sigma$. We use $\M(\sigma_{-S})$ to represent the result of running \ranking{} on $G$ restricted to the set of vertices $V\backslash S$  with respect to ordering $\sigma_{-S}$.

We will define the notion of \emph{backup vertex} of a vertex $u$ with respect to a permutation $\sigma_{-u^*}$ on $V \backslash \{u^*\}$. Intuitively, the backup vertex of a vertex $u$ with respect to some permutation is the second-best choice for $u$ in that permutation.

We will also define the notion of the \emph{profile} of a vertex $u$ in a rank vector $\vecx_{-u^*}$ on $V \backslash \{u^*\}$. The profile of $u$ in some permutation records the rank of $u$, the rank of $u$'s match (if it exists), and the rank of $u$'s backup (if it exists) in that permutation.

Notice that here we emphasize the permutation and rank vector to have domain $V \backslash \{u^*\}$. It is true that the definitions are well-defined even if we consider them as permutations and rank vectors on the full set of vertices $V$, but in our paper, we will mainly use these two notions when we assume $u^*$ is excluded.

\begin{definition}[Backup of a vertex $u$]\label{defbackup}
    Fix a permutation $\sigma_{-u^*}$ on $V\backslash\{u^*\}$ and assume $u$ has a match $v$ in $\M(\sigma_{-u^*})$. Let $\sigma_{-vu^*}$ be the result of removing $v$ from $\sigma_{-u^*}$. If $u$ is still matched in $\M(\sigma_{-vu^*})$, we call the match of $u$ in $\M(\sigma_{-vu^*})$ the backup of $u$ and denote it as $b$ with respect to permutation $\sigma_{-u^*}$. Moreover, if $u$ is not matched in $\M(\sigma_{-vu^*})$ we say $u$ has no backup in $\sigma_{-u^*}$.
\end{definition}

\begin{definition}[Profile of a vertex $u$]\label{def:profile} 
    Let $\vecx_{-u^*}$ be a rank vector on $V\backslash \{u^*\}$. The profile of $u$ in $\vecx_{-u^*}$ is a triple $(\x_u,\x_v,\x_b)$ where 
    \begin{enumerate}
        \item $\x_{u}$ is the rank of $u$ in $\vecx_{-u^*}$
        \item $\x_{v}$ is the rank of $v$ if $u$ is matched to $v$ in $\M(\sigma_{\vecx_{-u^*}})$; $\x_v=\bot$ if $u$ is not matched,
        \item $\x_{b}$ is the rank of $b$ where $b$ is the backup of $u$ in $\M(\sigma_{\vecx_{-u^*}})$; $\x_b=\bot$ if $u$ does not have a backup.
    \end{enumerate}
\end{definition}

\section{Structural Properties of \ranking{}}
In this section, we will mainly present a few structural properties of the \ranking{} algorithm. This section consists of three subsections, each of which will serve the following objective:

\paragraph{\cref{sectionalternatingpath}}: This subsection presents the alternating path lemma, a key property of \ranking{} that frequently appears in other related works. We prove this lemma in its strongest form, taking into account not only the final matching produced by \ranking{} but also the partial matchings generated during its execution. This lemma will be applied in most of the subsequent proofs in this section.

\paragraph{\cref{section3.3}}: This section describes the effects on the output of \ranking{} when a new vertex is introduced. In particular, we will focus on a prechosen pair of perfect matches $u,u^*$, discussing what happens when $u^*$ does not exist in the graph and what happens if $u^*$ is introduced.

\paragraph{\cref{section3.4}}: This section describes the effect of increasing the rank of a vertex's match in the permutation. We will prove that, under certain conditions, increasing the rank of $u$'s match does not affect the output matching, allowing us to derive a conditional monotonicity property for the rank of the $u$'s match.

In this section, we assume $\sigma:V\to\{1,...,|V|\}$ is directly drawn as a permutation instead of first picking the rank vector $\vecx$. Whenever we refer to the word rank of a vertex $u$, it refers to $\sigma(u)$, which is an integer and is the relative order of $u$ within the permutation. All the properties proved in this section are applicable to the rank vector form $\vecx$ as well. 

\subsection{Alternating Path Lemma}\label{sectionalternatingpath}
The alternating path lemma is a common result used in \ranking{} analysis. Different versions of it are proved in \cite{chan2018analyzing, Howtomatchwhenallverticesarriveonline,towardsabetterunderstanding}. The version we will prove and use is similar to Lemma 2.5 in \cite{Howtomatchwhenallverticesarriveonline}, but extends to partial matchings formed during the execution of \ranking{}.

Let $\sigma: V \to \{1, \dots, |V|\}$ be a permutation, and let $\sigma_{-u^*}$ be the induced permutation on $V \backslash \{u^*\}$\footnote{The alternating path lemma works for any graph $G=(V,E)$ and any $v\in V$, reguardless of wether $v$ is a vertex in the maximum matching,  we are choosing $u^*$ here because most of the application of alternating path lemma in this paper will pick $u^*$ as the special node.}. For the purpose of the proof, instead of treating $\sigma_{-u^*}$ as a permutation on $V \backslash \{u^*\}$, we interpret it as a permutation on $V$ but with $u^*$ marked as unavailable from the start of the algorithm, so that $u^*$ will never be included in the output matching. This perspective allows us to view \ranking{} running on $\sigma$ and $\sigma_{-u^*}$ as probing the same sequence of edges in the greedy probing view, as defined in \cref{greeedyprobing}, with the only difference occurring in the availability status of vertices.

The alternating path lemma states that if we take snapshots of \ranking{} running in parallel on $\sigma$ and $\sigma_{-u^*}$, the difference in their partial matchings forms an alternating path $(u_0, u_1, \dots, u_k)$ starting at $u^* = u_0$. Furthermore, the ranks of the even-indexed and odd-indexed vertices along the alternating path monotonically increase, that is, $\sigma(u_i) < \sigma(u_{i+2})$.

We introduce the notion of time $t$ to describe the alternating path lemma. Let $t \in \{1, \dots, |V|\}^2$ be a tuple, and let $\sigma$ be a permutation on $V$. We use time $t$ to denote the time when \ranking{} has just finished probing all edges $(u,v)$ with edge order $(\sigma(u), \sigma(v)) < t$. We use $\M^t(\sigma)$ to denote the partial matching formed by \ranking{} on $\sigma$ at this time, and $A^t(\sigma)$ to denote the set of vertices that are still available at this time.

To develop intuition for the alternating path lemma, consider executing the \ranking{} algorithm twice: once on the full graph, and once with a vertex $u^*$ removed. The resulting matchings differ only locally around $u^*$, as the remainder of the graph is unaffected. Crucially, the symmetric difference of the two matchings forms an alternating path, where edges alternate between the matchings and the path begins at $u^*$.

\begin{lemma}[Alternating Path Lemma]\label{lem:alt-path}
    Let $t\in \{1,...,|V|\}^2$ be a time stamp in the greedy probing process. The symmetric difference of $\M^t(\sigma)$ and $\M^t(\sigma_{-u^*})$, denoted as $\M^t(\sigma)\oplus \M^t(\sigma_{-u^*})$ is an alternating path $(u_0,u_1,..,u_k)$ where $u_0=u^*$ s.t.
\begin{enumerate}
    \item For all even $i < k$, $(u_{i}, u_{i+1}) \in \M^t(\sigma)$, for all odd $i < k$, $(u_{i}, u_{i+1}) \in \M^t(\sigma_{-u^*})$.
    \item $\sigma(u_i) < \sigma(u_{i+2})$ for all $i \leq k - 2$.
    \item $A^t(\sigma) \oplus A^t(\sigma_{-u^*}) = \{u_k\}$.
\end{enumerate}
\end{lemma} 
\begin{proof}
We prove this by induction on $t$.
\newline
\noindent
\begin{minipage}{0.6\textwidth}
\paragraph{\textbf{Base case:}} At $t=(1,1)$, no edges are probed yet. $\M^t(\sigma)\oplus \M^t(\sigma_{-u^*})$ is a degenerate path, which can be viewed as a path with a single node $u^*$. (1) and (2) hold trivially. By the special view we take on $\sigma_{-u^*}$, we have $A^{t}(\sigma)=V$ and $A^{t}(\sigma_{-u^*})=V\backslash\{u^*\}$, so (3) also holds.
\end{minipage}
\begin{minipage}{0.4\textwidth}
\begin{center}
\begin{tikzpicture}[dot/.style={circle, fill=black, inner sep=1.5pt}, scale=0.9]

\node[dot, label=right:{$u^{*}$}] (u0) at (1.5, 1.25) {};
\node[dot, label=right:$u_2$] (u2) at (1.5, 0.75) {};
\node[dot, label=left:$u_1$]  (u1) at (0, 0.25) {};
\node[dot, label=left:$u_3$]  (u3) at (0, -0.6) {};
\node[dot, label=right:$u_4$] (u4) at (1.5, -1) {};

\draw[thick, dash pattern=on 6pt off 3pt] (u0) -- (u1);
\draw[ thick] (u2) -- (u1);
\draw[thick, dash pattern=on 6pt off 3pt] (u2) -- (u3);
\draw[ thick] (u3) -- (u4);

\begin{scope}[shift={(2.5,0)}] 
  \draw[thick, dash pattern=on 6pt off 3pt] (0, 0.6) -- (0.6, 0.6);
  \node[anchor=west,font=\footnotesize] at (0.7, 0.6) {$\M^t(\sigma)$};

  \draw[ thick] (0, 0.2) -- (0.6, 0.2);
  \node[anchor=west,font=\footnotesize] at (0.7, 0.2) {$\M^t(\sigma_{-u^*})$};
\end{scope}

\end{tikzpicture}
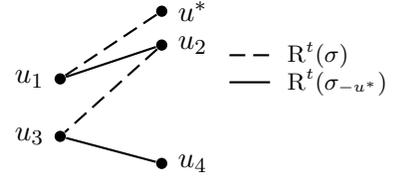
{\footnotesize\text{An example alternating path}}
\end{center}
\end{minipage}

\paragraph{\textbf{Inductive step:}} Assume it holds for $t$, let $t^+$ be the successive time stamp. If no edge is probed at time $t$, then we have nothing to prove. If there is some edge $(u,v)$ probed at time $t$, there are two cases:

\textbf{Case 1}: $(u,v)$ is included by \ranking{} in both $\M(\sigma)$ and $\M(\sigma_{-u^*})$, or not included in both $\M(\sigma)$ and $\M(\sigma_{-u^*})$. In this case, the symmetric difference $\M^{t^+}(\sigma)\oplus \M^{t^+}(\sigma_{-u^*})$ stays the same, as both either added an edge or both did not add anything. Similar reasoning holds for $A^{t^+}(\sigma) \oplus A^{t^+}(\sigma_{-u^*})$. Therefore, (1), (2), and (3) hold.

\textbf{Case 2}: $(u,v)$ is included in one of the partial matchings but not in the other. Let $\M^t(\sigma)\oplus \M^t(\sigma_{-u^*}) = (u_0, \dots, u_k)$. Case 2's assumption implies that one of $u$ or $v$ is $u_k$, because $u_k$ is the only vertex that can have a different availability status across $A^t(\sigma)$ and $A^t(\sigma_{-u^*})$. W.o.l.g., let $u = u_k$. 
We only prove the case where $k$ is even, as the case where $k$ is odd is symmetric. By (1), the alternating path ends at an edge in $\M^t(\sigma_{-u^*})$, meaning $u$ and $v$ are available in $A^t(\sigma)$ and hence will be added to $\M^{t^+}(\sigma)$. However, $(u,v)$ will not be added to $\M^{t^+}(\sigma_{-u^*})$ as $u$ is not available. Therefore, the new symmetric difference of partial matchings is an alternating path $(u_0, \dots, u_k = u, v)$ with $(u,v) \in \M(\sigma)$. This proves (1). By case 2's assumption and (3), we have $u,v \in A^t(\sigma)$ and $v \in A^t(\sigma_{-u^*})$. Now, $u, v$ are removed from $A^t(\sigma)$, but $v$ remains available in $A^t(\sigma_{-u^*})$. Hence, $A^{t^+}(\sigma)\oplus A^{t^+}(\sigma_{-u^*}) = \{v\}$, which proves (3). To prove (2), observe that $(u_{k-1}, u)$ is probed before $(u,v)$. By \cref{lexicographicorder}, this implies $\sigma(u_{k-1}) < \sigma(v)$.
\end{proof}

\subsection{Structural Properties Comparing $\M(\sigma)$ and $\M(\sigma_{-u^*})$}\label{section3.3}

Let $\sigma_{-u^*}$ be a permutation on $V\backslash\{u^*\}$, and let $\sigma$ be a permutation formed by introducing $u^*$ into $\sigma_{-u^*}$. In this subsection, we will prove the following properties:
\begin{enumerate}
    \item If $u$ is matched with $v$ when $u^*$ is excluded, then inserting $u^*$ into a rank  $<\sigma(v)$ will make $u^*$ match to a vertex of rank $\leq\sigma(u)$.
    \item If $u$ is matched with $v$ when $u^*$ is excluded, then inserting $u^*$ into a rank $>\sigma(u)$ will not cause $u$ to be unmatched or matched to a less preferred vertex than $v$.
    \item If $u$ is is matched with $v$ when $u^*$ is excluded, and introducing $u^*$ causes $u$ to be unmatched or matched to a less preferred vertex than $v$, then $u^*$ is matched with a vertex of rank $\leq\sigma(v)$.
    \item If $u$ admits a backup $b$ in $\sigma_{-u^*}$ as defined in \cref{defbackup}, then $u$ will always match to a vertex with rank $\leq\sigma(b)$, no matter where $u^*$ is inserted.
    \item If $u$ is unmatched when $u^*$ is excluded, then $u^*$ will always match to a vertex of rank $\leq\sigma(u)$, no matter where $u^*$ is inserted.
\end{enumerate}
The following claim proves the first three properties.
\begin{claim}\label{lem:struct-property}
    Let $\sigma_{-u^*}$ be a permutation on $V\backslash\{u^*\}$. Let $\sigma$ be a permutation generated by inserting $u^*$ to an arbitrary rank in $\sigma_{-u^*}$. Assume $u$ is matched to $v$ in $\M(\sigma_{-u^*})$. We have the following:
\begin{enumerate}
    \item If $\sigma(u^*)<\sigma(v)$, then $  u^\ast$ is matched to some vertex $w$ in $\M(\sigma)$ where $\sigma(w)\leq \sigma(u),$
   \item If $\sigma(u^*)>\sigma(u)$ , then $u$ is matched to some vertex $w$ in $\M(\sigma)$ where $\sigma(w)\leq \sigma(v),$
    \item If $u$  is not matched in  $\M(\sigma)$ or the match of $u$ in $\M(\sigma)$ has rank $>\sigma(v)$, then  $u^*$  is matched to some vertex $w$ where $ \sigma(w)\leq \sigma(v).$
\end{enumerate}
\end{claim}
\begin{proof}
    (1) Assume $\sigma(u^*)<\sigma(v)$. At the time $t$ when $(u,u^*)$ is probed, either $u^*$ is already matched to some vertex $w$ or it is not. If $u^*$ is matched to $w$, then $(w,u^*)$ being probed earlier than $(u,u^*)$ implies $\sigma(w)\leq \sigma(u)$ by \cref{lexicographicorder}. If $u^*$ is still available, by part (3) of \cref{lem:alt-path}, the only vertex with a different availability status between $A^t(\sigma)$ and $A^t(\sigma_{-u^*})$ is $u^*$. Since $(u,v)$ has not been probed yet at time $t$ (by \cref{lexicographicorder}), $u\in A^t(\sigma_{-u^*})$ and is also available in $A^t(\sigma)$. Since both $u$ and $u^*$ are available, $u^*$ will be matched to $u$ in $\M(\sigma)$.

(2) Assume $\sigma(u^*)>\sigma(u)$, and towards a contradiction, assume that $u$ is not matched in $\M(\sigma)$ or is matched to some vertex $w$ with rank $\sigma(w)>\sigma(v)$ in $\M(\sigma)$. It follows that $(u,v)\in \M(\sigma)\oplus \M(\sigma_{-u^*})$. Denote the alternating path as $(u_0,u_1,\dots,u_k)$. By property (1) in \cref{lem:alt-path}, it follows that a node in the alternating path is even-indexed if and only if it is unmatched in $\M(\sigma)$ or its match in $\M(\sigma)$ has a larger rank than its match in $\M(\sigma_{-u^*})$. By our previous assumption, $u$ is forced to take an even-indexed position in the alternating path. But monotonicity implies the even-indexed $u$ must satisfy $\sigma(u)>\sigma(u^*)$, which is a contradiction.

(3) If $u$ is not matched in $\M(\sigma)$, by the same line of reasoning as in case (2), $u$ is an even-indexed vertex in the alternating path. It follows that $v$ is an odd-indexed vertex. By the monotonicity of the alternating path, the match of $u^*$ in $\M(\sigma)$, denoted as $w$, is the first odd-indexed vertex and thus has the smallest rank among all the odd-indexed vertices. This implies that $\sigma(w) \leq \sigma(v)$.
\end{proof}
\cref{lem:struct-property} describes how the insertion of a vertex $u^*$ into permutation $\sigma_{-u^*}$ affects the matching outcome. Intuitively, if $u^*$ is inserted earlier in the permutation than $u$'s original match $v$, then $u^*$ is preferred by $u$ more than $v$. As a result, $u$ will always pick $u^*$ if $u^*$ is still not matched until the time $(u,u^*)$ gets probed. If $u^*$ is placed before $u$, it is preferred by $v$ over $u$, so it may be matched to $v$ first, thereby altering $u$'s match. Conversely, if $u^*$ is inserted later than $u$, it will not take $v$ away from $u$, as $v$ always prefers $u$ over $u^*$. This claim formalizes how the position of $u^*$ relative to $u$ and $v$ in the arrival order governs the impact of $u^*$ on the matching structure. 

The next claim describes the case where $u$ has a backup $b$ in $\sigma_{-u^*}$. When $u$ has a match $v$ and admits a backup $b$ in $\M(\sigma_{-u^*})$, the following two properties hold: (1) $v$ has a smaller rank than $b$. (2) If $u$ becomes worse off when $u^*$ is added to $\sigma_{-u^*}$, the worst-case match $u$ will get is $b$.
\begin{claim}\label{lem:backup}
Assume $u$ has a backup vertex $b$ in $\sigma_{-u^*}$ as defined in \cref{defbackup}. Let $\sigma$ be a permutation generated by inserting $u^*$ to an arbitrary rank in $\sigma_{-u^*}$. We have the following: 
\begin{enumerate}
    \item $\sigma_{-u^*}(v)<\sigma_{-u^*}(b),$
    \item $u$ is always matched to some vertex $w$ with rank $\sigma(w)\leq \sigma(b)$ in $\M(\sigma)$.
\end{enumerate} 
\end{claim} 
\begin{proof} 
\textbf{Proof for (1)}: We apply the alternating path lemma to $\sigma_{-u^*}$ and $\sigma_{-vu^*}$. As $v$ is the vertex equivalent to $u^*$ in \cref{lem:alt-path}, the alternating path will start with $v=u_0$ and consequently $u=u_1$. By the definition of $b$, $(u,b)\in \M(\sigma_{-vu^*})$ implies $b=u_2$. By monotonicity (property (2) of the alternating path lemma), we have $\sigma_{-u^*}(v)<\sigma_{-u^*}(b)$.

\textbf{Proof for (2)}: Similar to what we did in the alternating path lemma, we view \ranking{} on $\sigma$, $\sigma_{-u^*}$, $\sigma_{-vu^*}$ as probing the same lists of elements. In the list for $\sigma_{-u^*}$, $u^*$ is marked as unavailable from the beginning; in the list for $\sigma_{-vu^*}$, both $v,u^*$ are marked as unavailable from the beginning. Let $t$ be the time when $(u,v)$ is probed in $\sigma$. There are three cases.

Case 1: $u$ is not available in $A^t(\sigma)$. If $u$ is not available, then it is matched to some $w$ already in $\M^t(\sigma)$. By $\cref{lexicographicorder}$, $(u,w)$ being probed earlier than $(u,v)$ implies $\sigma(w)<\sigma(v)$; combining with the fact that $\sigma(v)<\sigma(b)$ (which follows from $\sigma_{-u^*}(v)<\sigma_{-u^*}(b)$), we have $\sigma(w)<\sigma(b)$.

Case 2: Both $u,v$ are available in $A^t(\sigma)$. Then $(u,v)$ will be included in $\M(\sigma)$ and by fact (1), $\sigma(v)<\sigma(b)$ proves (2). 

\begin{figure}[h]
\centering
\begin{tikzpicture}[dot/.style={circle, fill=black, inner sep=1.5pt}, scale=0.9]

\node[dot, label=right:{$u^{*}$}] (u0) at (1.5, 1.5) {};
\node[dot, label=left:$u_1$] (u1) at (0, 0.75) {};
\node[dot, label=right:$u_2$]  (u2) at (1.5, 0.25) {};
\node[dot, label=left:{$u_k=v$}]  (u3) at (0, -1.5) {};
\node[dot, label=right:{$u_{k-1}$}] (u4) at (1.5, -0.75) {};

\draw[thick, dash pattern=on 6pt off 3pt] (u0) -- (u1);
\draw[thick] (u2) -- (u1);

\draw[dotted,thick] (0.75,0.25) -- (0.75,-0.75);

\draw[thick, dash pattern=on 6pt off 3pt] (u3) -- (u4);

\begin{scope}[shift={(2.5,0)}]
  \draw[thick, dash pattern=on 6pt off 3pt] (0, 0.6) -- (0.6, 0.6);
  \node[anchor=west] at (0.7, 0.6) {\footnotesize$\M^t(\sigma)$};

  \draw[thick] (0, 0.2) -- (0.6, 0.2);
  \node[anchor=west] at (0.7, 0.2) {\footnotesize$\M^t(\sigma_{-u^*})$};
\end{scope}

\begin{scope}[xshift=7cm]
\node[dot, label=right:{$v$}] (u0r) at (1.5, 0) {};
\end{scope}

\node at (1.5, -2) {\footnotesize\text{$\M^t(\sigma)\oplus\M^t(\sigma_{-u^*})$}};
\node at (8.5, -2) {\footnotesize\text{$\M^t(\sigma_{-u^*})\oplus\M^t(\sigma_{-vu^*})$}};
\end{tikzpicture}
\centering

{\centering
  \scriptsize Figure 1: Case 3\footnotemark\par
}
\label{figure1}
\end{figure}
\footnotetext{Figure 1 shows the alternating paths at time $t$ for Case 3. The alternating path between $\sigma$ and $\sigma_{-u^*}$ ends at $v$; the alternating path between $\sigma_{-u^*}$ and $\sigma_{-vu^*}$ is a degenerated path.}
Case 3: $u$ is available, but $v$ is not in $A^t(\sigma)$. In this case, $v \notin A^t(\sigma)$ but $v \in A^t(\sigma_{-u^*})$, since $(u,v)$ is included in $\M(\sigma_{-u^*})$. By Fact (3) in the alternating path lemma, we have $A^t(\sigma_{-u^*}) = A^t(\sigma)\cup \{v\}$. Notice that it is also true that $v \notin A^t(\sigma_{-vu^*})$, as it is marked as unavailable from the beginning. This implies $A^t(\sigma_{-u^*}) = A^t(\sigma_{-vu^*})\cup \{v\}$. Together, the above two facts imply that $A^t(\sigma_{-vu^*}) = A^t(\sigma)$. From this point on, the remaining run of \ranking{} on $\sigma$ and $\sigma_{-vu^*}$ will probe the same edge list and always include the same edges, as the set of available vertices will always remain the same. In particular, since $(u,b)$ is included in $\M(\sigma_{-vu^*})$, it will also be included in $\M(\sigma)$, which shows that $u$ is matched to a vertex of rank $\leq \sigma(b)$.
\end{proof}

Next is a simple fact for the case where $u$ is not matched before $u^*$ is introduced. In this case, if $u^*$ is not matched until $(u,u^*)$ is probed, they will be matched together.
\begin{fact}\label{lem:nomatch}
    Let $\sigma$ be a permutation generated by inserting $u^*$ to an arbitrary rank in $\sigma_{-u^*}$. If $u$ is not matched in $\M(\sigma_{-u^*})$, then $u^*$ is always matched to some vertex $w$ of rank $\sigma(w)\leq \sigma(u)$ in $\M(\sigma)$.
\end{fact}
\begin{proof}
Let $t$ be the time when $(u, u^*)$ is probed. If $u^*$ is already matched to some vertex $w$, then by \cref{lexicographicorder}, the fact that $(w,u^*)$ is probed earlier than $(u,u^*)$ implies that $\sigma(w) < \sigma(u)$. If $u^*$ is not matched, by alternating path lemma (3), we have $A^t(\sigma) \oplus A^t(\sigma_{-u^*}) = \{u^*\}$. Since $u$ is not matched in $\M(\sigma_{-u^*})$, this implies $u$ is available at time $t$ in $A^t(\sigma_{-u^*})$, and hence also available in $A^t(\sigma)$. Since both $u$ and $u^*$ are available in $A^t(\sigma)$, \ranking{} will include $(u,u^*)$ in the matching.
\end{proof}
\subsection{Properties of Profiles}\label{section3.4}
In this section, we will prove two main properties for the distribution of $u$'s match $v$ when $u^*$ is excluded. The two main properties will entail monotonicity properties for the profiles of $u$. Recall the definition of profile from \cref{def:profile}. We will prove that, for any fixed $\x_u$:

\begin{enumerate}
\item Conditioning on $u$ not having a backup, the conditional distribution of $\x_v$ has a monotonically increasing probability mass(density) function in range $(0,1)$.
\item Conditioning on $u$ having a backup at $\x_b$, the conditional distribution of $\x_v$ has a monotonically increasing probability mass(density) function in range $(0,x_b)$.
\end{enumerate}

It is easier to argue with respect to the induced permutation $\sigma_{\vecx_{-u^*}}$, so we will be working with permutations $\sigma_{-u^*}$ instead of rank vectors in this section. The properties for $\vecx_{-u^*}$ follow as easy corollaries. We define the following notation, which will be important in proving the next few lemmas:

 \begin{definition}
     Let $\sigma$ be a permutation. Define \(\sigma_v^i\) as the permutation formed by moving \(v\) to rank \(i\) and shifting all other vertices accordingly. Formally, \(\sigma_v^i\) is the unique permutation such that \(\sigma_v^i(v) = i\), and for all vertices \(u, w \) other than $v$, 
$\sigma(u) \leq \sigma(w) \iff \sigma_v^i(u) \leq \sigma_v^i(w).$
 \end{definition}
As the above notation is heavy, instead of using $\sigma_{-u^*}$ as the permutation on $V \backslash \{u^*\}$, we will simply use $\sigma$ in the rest of this section. This avoids the use of notations like $(\sigma_{-u^*})_v^i$ to increase readability.

 Before proving the two main points of this section, we will prove a useful fact that will be helpful later. Let $t$ be the time $v$ becomes unavailable during the \ranking{} process on $\sigma$. In this section, we will use time $t$ as an integer in $\{1, \dots, |V|-1\}$, meaning we are adopting the vertex iterative view of \ranking{} rather than the greedy probing view. Formally, $t \in \{1, \dots, |V|-1\}$ represents the time when \ranking{} has just finished processing the vertices of rank $<t$, it corresponds to time $(t,1)$ in the greedy probing view. Similar to the proof of the alternating path lemma \cref{lem:alt-path}, we will assume in $\sigma_{-v}$ that \ranking{} processes $v$ with $v$ marked as unavailable, to make sure \ranking{} $\sigma$ and $\sigma_{-v}$ process the same node at the same time.
 
 We will prove a useful fact first. It roughly states that increasing the rank of $v$ by 1 in $\sigma$ or removing $v$ from $\sigma$ will not have any effect on the partial matching or the set of available vertices (except $v$) before the time $v$ becomes unavailable in $\sigma$.

Let $\sigma$ be a permutation over $V\backslash\{u^*\}$. Let $\sigma_{-v}$ be the induced ordering where $v$ is removed. Let $\sigma^+=\sigma^{\sigma(v)+1}_v$ be the permutation that swaps $v$ with the vertex that comes right after $v$, denote as $v'$. Let $\M^t(\sigma),A^t(\sigma)$ be the respective partial matching and available vertex set for $\sigma$ at time $t$.

\begin{fact}\label{same-match-before-v}
For any time $t$, if $v$ is still available at time $t$ with respect to the \ranking{} process on $\sigma$, i.e. $v\in A^t(\sigma)$, then the partial matchings $\M^t(\sigma), \M^t(\sigma^+), \M^t(\sigma_{-v})$ are the same, and the sets of available vertices $A^t(\sigma), A^t(\sigma^+), A^t(\sigma_{-v})$ are the same except at $v$.
\end{fact}
\begin{proof}
This result establishes that if a vertex $v$ remains unmatched up to a certain point in the permutation, then demoting its rank or removing it does not alter the state of the partial matching before it becomes unavailable. Intuitively, since $v$ has not been selected by any vertex thus far, it has exerted no influence on the outcome, and hence making $v$ less preferred or even removing it does not affect the partial matching. See the detailed proof in \nameref{Appendix1}. 
\end{proof}

We now start proving the main results for this section. The next lemma shows that increasing the rank of $v$ by $1$ has no effect on the matching when $u$ does not have a backup vertex $b$.
\begin{lemma}\label{descrete-monotonicity-nobackup} 
Let $\sigma$ be an permutation on $V\backslash\{u^*\}$ where $u$ is matched with $v$ in $\M(\sigma)$ and $u$ does not have a backup. Let $\sigma^+ = \sigma_v^{\sigma(v) + 1}$ be the permutation formed by increasing the rank of $v$ by $1$, in other words, swapping $v$ with the next vertex $v'$. Then $u$ is still matched to $v$ in $\M(\sigma^+)$, and $u$ still does not have a backup.

As a result, for $\sigma(v) \leq i \leq |V|-1$, $u$ is still matched to $v$ in $\M(\sigma^i_v)$ and $u$ still does not have a backup.
\end{lemma}
\begin{proof}
Let $t$ be the time when $v$ becomes unavailable in the \ranking{} process on $\sigma$. It's easy to see $t$ is the minimum of $\sigma(u)$ and $\sigma(v)$. We consider three different cases:

Case 1: $t=\sigma(u) < \sigma(v)$. Then $t$ is the time when $u$ picks $v$. Since $v$ is still available at this time, \cref{same-match-before-v} implies $\M^t(\sigma) = \M^t(\sigma^+)$ and ${u, v} \subseteq A^t(\sigma) = A^t(\sigma^+)$. It suffices to show that $v$ is the unique available neighbor of $u$ at this time. Since $\M^t(\sigma) = \M^t(\sigma_{-v})$ also holds, any available neighbor other than $v$ in $A^t(\sigma)$ would also appear in $A^t(\sigma_{-v})$, contradicting the assumption that $u$ is not matched in $\M(\sigma_{-v})$.

Case 2:  $t=\sigma(v) = \sigma(u) - 1$. Here, $\sigma^+$ simply swaps the order of $u$ and $v$. By \cref{same-match-before-v}, $u$ and $v$ are both available in $A^t(\sigma^+)$. At time $t$, $u$ will actively pick its smallest ranked neighbor according to $\sigma^+$. Since $\sigma^+(v) = \sigma^+(u) + 1$, $v$ is necessarily the smallest ranked vertex among the available neighbors of $u$, and will be picked by $u$.

Case 3: $t=\sigma(v) < \sigma(u) - 1$. By \cref{same-match-before-v}, we have $A^t(\sigma) = A^t(\sigma^+)$. We will show that the vertex being processed at time $t$ by \ranking{} on $\sigma^+$, namely $v'$, will not pick $u$ or $v$ as its match. If this is the case, then $v$ will be the next vertex picking and has $u$ as an available neighbor. Since the relative order of all vertices $w$ with rank $\sigma(w)> \sigma(v)+1$ is not changed between $\sigma$ and $\sigma^+$, and all available neighbors of $v$ at the time $t+1$ fall in this range (because $u$ was the best choice for $v$ and its rank in $\sigma$ is $>t+1$), $v$ will still pick $u$ as its match. 

Suppose $v'$ is not matched yet so that it might pick $u$ or $v$ to match. First, we claim that $v'$ is not a neighbor of $v$. If $v$ and $v'$ were neighbors, then $v' \in A^t(\sigma^+) = A^t(\sigma)$ would be the best choice for $v$ in $\sigma$ when $v$ picks. This contradicts the assumption that $v$ picks $u$ as its match in $\M(\sigma)$.

On the other hand, Assume $v'$ picks $u$. then $u$ must be the best neighbor for $v'$ when it picks in $A^t(\sigma^+)$. By \cref{same-match-before-v}, the set of neighbors available in $A^t(\sigma_{-v})$ is also the same. Furthermore, the relative order of the set of available neighbors for $v'$ remains unchanged between $\sigma^+$ and $\sigma_{-v}$, so $v'$ would also pick $u$ in $\M(\sigma_{-v})$. This contradicts the assumption that $u$ does not have a backup.

Finally, we need to show that $u$ still does not have a backup in $\sigma^+$. This follows from the fact that $\sigma_{-v}$ and $(\sigma^+)_{-v}$ are the same ordering on the same induced graph, which means that \ranking{} will generate the same matching.
\end{proof}
\begin{corollary}\label{descretecoro-mootone-nobackup}
Let $\vecx$ be a rank vector on $V \backslash \{u^*\}$ where $u$ is matched with $v$ in $\M(\sigma_{\vecx})$ and $u$ does not have a backup. For $\x_v \leq t < 1$, let $\vecx'$ be the rank vector formed by increasing $\x_v$ to $t$ in $\vecx$. Then $u$ is still matched to $v$ in $\M(\sigma_{\vecx'})$ and still does not have a backup.
\end{corollary} 
\begin{proof} 
 The real-valued version follows as $\sigma_{\vecx'} = (\sigma_{\vecx})_v^i$ for some $i$ satisfying $\sigma_{\vecx}(v) \leq i \leq |V| - 1$ and \cref{descrete-monotonicity-nobackup}.
\end{proof}
The next lemma depicts that monotonicity property when $u$ has a backup $b$.
\begin{lemma}\label{descrete-monotonicity-backup} Let $\sigma$ be a permutation such that $u$ is matched with $v$ in $\M(\sigma)$ and $u$ has a backup $b$ at rank $\sigma(b)>\sigma(v)+1$, then $u$ is matched with $v$ in $\sigma^{+}=\sigma_v^{\sigma(v)+1}$ and $u$ has the same backup $b$ in $\sigma^+$.

As a result, for $\sigma(v) \leq i < \sigma(b)$, $u$ is still matched to $v$ in $\M(\sigma^i_v)$ and $u$ still has a backup $b$ at $\sigma^i_v(b) = \sigma(b)$
\end{lemma}
Proof is similar to the case \cref{descrete-monotonicity-nobackup}. See \nameref{Appendix1}

\begin{corollary}\label{descretecoro-mootone-backup}
Let $\vecx$ be a rank vector on $V \backslash \{u^*\}$ where $u$ is matched with $v$ in $\M(\sigma_{\vecx})$ and has a backup $b$ at $x_b$. For $x_v \leq t < x_b$, let $\vecx'$ be the rank vector formed by increasing $x_v$ to $t$ in $\vecx$. Then $u$ is still matched to $v$ in $\M(\sigma_{\vecx'})$ and still has a backup $b$ at $x'_b = x_b$.
\end{corollary} 
\begin{proof} 
 The real-valued version follows from the fact that $\sigma_{\vecx'} = (\sigma_{\vecx})_v^i$ for some $i$ satisfying $\sigma_{\vecx}(v) \leq i < \sigma_{\vecx}(b)$ and \cref{descrete-monotonicity-backup}.
\end{proof}
\section{Bucketed Random Permutation Generation}\label{def:BucketedRPG}
In the remainder of this paper, we take a slightly different approach to generating the random permutation for \ranking{}. This model is designed to approximate the uniform real-valued rank view.

Ideally, in real-valued permutation generation, gain-sharing would be performed using these real-valued ranks $\vecx$. However, in our LP formulation, the gain-sharing function is itself a parameter to be optimized, which prevents us from searching over infinitely many functions. For instance, we cannot search through all continuous functions on $(0,1)$, as there are uncountably many of them.

To address this, we discretize the interval $(0,1)$ into $k$ equal-sized intervals of width $\frac{1}{k}$, which we refer to as buckets. We then restrict our search space to piecewise functions that remain constant within each bucket. As a result, each vertex independently falls into one of these buckets with uniform probability $\frac{1}{k}$. It is straightforward to see that as $k$ approaches infinity, this discrete model approaches the continuous real-valued rank distribution. As a consequence, we will also use $x_v$ to represent the bucket value of a vertex $v$, taking values in $\{1,..,k\}$, to represent this fact\footnote{We will not directly address the real-valued ranks $\x_v$ from this section onward, so there should be no concerns about confusion regarding the notation. similarly, when we use the word \emph{rank vector}, it would refer to the bucketed rank vector defined in this section.}.

We define the bucketed random permutation generation model as the following:

Assume there are a total of $|V|$ vertices, and fix an arbitrary positive integer $\numpieces$ independent of $|V|$ representing the number of buckets. The final ordering is decided by two rounds of random assignments. During the first round, each vertex $v$ independently and uniformly selects a bucket $\x_v \in \{1, 2, \dots, \numpieces\}$. After the first round, if multiple vertices fall into the same bucket, we take a uniform random permutation of the vertices within each bucket.

 The rank of each vertex $v$ can be represented by two numbers, $(\x_v, \y_v)$, where $\x_v$ is the assignment of the bucket and $\y_v$ is the rank of $v$ within the bucket, based on the local random permutation. The final order is determined using the lexicographic ordering of $(\x_v, \y_v)$. Formally, for any two vertices $u$ and $v$, the ordering is defined as follows:
\[
(\x_u, \y_u) < (\x_v, \y_v) \Longleftrightarrow \text{ (1) } \x_u < \x_v \text{ or (2) } \x_u = \x_v \text{ and } \y_u < \y_v.
\]
For now, we use vector $\vecx,\vecy$ to represent the vector of values $\{\x_v\}_{v\in V}$ $\{\y_v\}_{v\in V}$. We also denote $\sigma_{\vecx, \vecy} : V \to \{1, \dots, |V|\}$ as the permutation induced by $\vecx, \vecy$. $\sigma_{\vecx, \vecy}$ can be formally defined as the unique permutation $\sigma' : V\to \{1,...,|V|\}$ such that $ \sigma'(u) \leq \sigma'(v) \Leftrightarrow (\x_u, \y_u) \leq (\x_v, \y_v)$.

\begin{claim}\label{bucketedpermutation}
Bucketed random permutation ordering on $G=(V,E)$ gives rise to the same distribution as picking a uniform random permutation $\sigma:V\to \{1,...,|V|\}$. That is, for each fixed permutation $\sigma_0$,
$$\pr(\sigma_{\vecx,\vecy}=\sigma_0)=\frac{1}{|V|!}$$
\end{claim}
\begin{proof}
    See \nameref{Appendix1}.
\end{proof}
Further, all the structural properties we proved in the previous section that hold for a permutation $\sigma:V\to|V|$ also hold for the bucketed rank vectors $\vecx,\vecy$, as they are directly applicable to the induced permutation $\sigma_{\vecx,\vecy}$\footnote{All structural properties proved in Section~4 (besides the monotonicity properties) only discuss the relative ordering of vertices instead of absolute ordering. Since $\sigma_{\vecx,\vecy}(u)$ and $(x_u, y_u)$ admit the same ordering, the structural properties are directly transferable. As for the monotonicity property, \cref{descrete-monotonicity-nobackup} and \cref{descrete-monotonicity-backup} are presented in a form that only involves the relative orderings as well and hence are also transferable; we will use the lemmas in this form and discuss the related probability distributions in \cref{section8}. In fact, it will turn out that for bucketed permutations, if $u$ has a backup $b$ in bucket $\x_b$, the monotonicity property only holds up to bucket $\x_b-1$ instead of $\x_{b}$; this causes us to lose a small fraction in the approximation ratio.}. Because of this fact, we will use the same notation $\sigma$ to represent the rank vector, i.e., $\sigma(u) = (\x_u, \y_u)$.

The formal definition of $\sigma_{-u^*}$, $\sigma_{v}^{(x,y)}$ and $\sigma$  as bucketed rank vectors are the following:
\begin{definition}\label{defofpermutations}
    $\sigma_{-u^*}:V\backslash\{u^*\} \to \{1, \dots, k\} \times \{1,..,|V|-1\}|$ is a valid bucketed rank vector if it is injective and the range for any fixed bucket is a contiguous integer interval starting at 1, i.e., for a fixed $i \in \{1, \dots, k\}$, we have $\{y \mid (i, y) \in \sigma(V)\} = \{1, \dots, m\}.$ We will use $\x_u / \y_u$ to represent the relative first/second index of $\sigma(u)$.

    $\sigma_{v}^{(x,y)}$ is well defined if it is the bucketed rank vector formed by moving/inserting $v$ (depending on if $v$ is in the domain of $\sigma$ or not) to rank $(x,y)$ and shifting all other vertices accordingly without affecting the bucket values and the relative orders. That is, $\sigma_{v}^{(x,y)}$ is the unique bucketed rank vector $\sigma'$ such that $\sigma'(v) = (x, y)$, and for all $ u, w \neq v$, $\x_u = \x'_u$ and $\sigma(u) \leq \sigma(w) \Leftrightarrow \sigma'(u) \leq \sigma'(w)$. 

    Let $\sigma: V \to \{1, \dots, k\} \times \{1, \dots, |V|\}$ be the bucketed rank vector obtained by inserting $u^*$ into $\sigma_{-u^*}$ at a specific position $(x,y)$, shifting all vertices within the same bucket accordingly. Note that the notation $\sigma_{v}^{(x,y)}$ is well-defined no matter if $v$ is in the domain of $\sigma$ or not, allowing us to define $\sigma$ as $(\sigma_{-u^*})_{u^*}^{(x,y)}$. Conversely, we can view $\sigma_{-u^*}$ as the permutation induced from $\sigma$ by removing $u^*$ while preserving bucket values and relative orders.
\end{definition}

\section{Two-Patitioning Based Gain Sharing}\label{GainSharingMethod}
To analyze the expected size of $\M(\sigma)$, we use the notion of gain sharing (randomized primal-dual method).

For each edge $(u,v)$ selected by \ranking{}, we assign a total gain of $1$ and split it among the incident nodes $u$, $v$. Denote the gain of node $u$ as $\g(u)$. It is not hard to see that we have
$$
\mathbb{E}\left[|\M(\sigma)|\right] =  \mathbb{E}\left[\sum_{u \in V}\g(u)\right] = \sum_{(u,u^*) \in M^*} \mathbb{E}\left[\g(u) + \g(u^*)\right],
$$
where we assumed $M^*$ is a perfect matching for $G$.  It turns out to get a lower bound for the approximation ratio, it suffices tho bound $\ev[g(u)+g(u^*)]$.
\begin{claim}\label{guustarisalowerbound}
    Let $\alpha\in[0,1]$ be a value s.t. for all pairs of perfect match $(u.u^*)$ in $G$,
    $$\alpha\leq \ev[g(u)+g(u^*)].$$
    Then $\alpha$ lower bounds the approximation ratio for \ranking{} on $G$.
\end{claim}
\begin{proof}
    This follows from the following calculation, assuming $|V|=2n$ and $G$ has a perfect matching $M^*$:
    $$\alpha\leq\frac{n\cdot\ev[g(u)+g(u^*)]}{n} =\frac{\sum_{(u,u^\ast)\in M^*}\ev [\g(u)+\g(u^\ast)]}{n}= \frac{\ev [|\M(\sigma)|]}{\ev [|M^*|]}.$$ 
\end{proof}
An example of this type of analysis is presented in \cite{eden2020economicbasedanalysisrankingonline}. In this paper, the analysis is conducted on a bipartite graph. The author assigned the two sides of the bipartition to play the role of buyers ($\B$) and items ($\I$). An auxiliary price function $\price$ is defined for the items. Whenever a match occurs between a buyer $b$ and an item $i$, the intuition is that the buyer $b$, possessing $1$ unit of budget, buys the item $i$ with price $\price(\x_i)$. The result is item $i$ receiving a utility of $g(i)=\price(i)$ and buyer $b$ receiving a utility of $g(b)=1 - \price(\x_i)$. Then, as in \cref{guustarisalowerbound}, expected value of $g(u)+g(u^*)$ is taken to obtain a lower bound for \ranking{}. 

The function $g$ that results in a good approximation ratio is asymmetric with respect to buyers and items, which means that we cannot apply this approach of sharing gains if two vertices of the same group are matched. However, it is impossible to pick a fixed bipartition for general graphs that does not result in a match within the same group. To address this issue, we introduce a randomized bipartition scheme. The scheme we introduced satisfies the following three properties, which are enough for us to conduct a buyer-item style analysis. 
\begin{enumerate}
    \item  Whenever a vertex $u$ is matched to a vertex $v$ by \ranking{}, $u$ and $v$ belong to different sides of the bipartition.
    \item For all pairs of perfect matches $u, u^*$, $u$ and $u^*$ belong to different sides of the bipartition.
    \item For any vertex $v$, conditioning on the event that $v$ belongs to a fixed group (either buyer or item) does not change the underlying distribution of $\sigma$.
\end{enumerate}

The construction is as follows: We construct a random coloring function $\chi: V \to \{B, I\}$ and group vertices having the same colors together, i.e., buyers are the set of vertices such that $\chi(v) = B$ and items are the set of vertices with $\chi(v) = I$. For each $\sigma$ and its respective result matching $\M(\sigma)$, since $\M(\sigma) \cup M^*$ does not contain odd cycles, we can pick a two-coloring $\chi_\sigma: V \to \{B, I\}$ on $\M(\sigma) \cup M^*$. Whenever $\sigma$ is realized, we let $\chi=\chi_\sigma$ with probability $\frac{1}{2}$ and let $\chi$ equals the opposite coloring of $\chi_\sigma$ with probability $\frac{1}{2}$.

Because we choose $\chi_{\sigma}$ as a two coloring on $\M(\sigma)\cup M^*$, properties 1 and 2 follows. On the other hand, because we choose $\chi_{\sigma}$ and its reversed coloring each with $\frac{1}{2}$ probability, any vertex $v$ has $\frac{1}{2}$ probability of being assigned to group $B$ and $\frac{1}{2}$ probability of being assigned to group $I$, independent of the underlying distribution of $\sigma$.
\subsection{Randomized Two-partitioning on $V$}
In this section, we present the formal definition of the randomized two-partioning $\chi$ and prove that the three desired properties hold. We first state the following fact about the union of two matchings:
\begin{fact}\label{altpathsandcycles}
    Let $M_1$, $M_2$ be two matchings on $G$. Then $M_1\cup M_2$ is a disjoint union of alternating paths and even cycles.
\end{fact}
Now we define the random bipartition as a random coloring and prove that the desired properties hold.
\begin{definition}(Randomized Two-Partitioning)\label{def:2-coloring}
    For each rank vector $\sigma$ on $V$ and its respective matching $\M(\sigma)$. Let $\chi_{\sigma}:V\to \{\B,\I\}$ be a fixed two-coloring on $\M(\sigma)\cup M^*$. $\chi_{\sigma}$ always exists because \cref{altpathsandcycles}. Let $\chi_{\sigma}^r$ be the opposite coloring of $\chi_{\sigma}$, i.e. all vertices have their color flipped. We define the random two-partitioning $\chi$ to be a random variable s.t. $p_{joint}$, the joint distribution of $\chi$ and $\sigma$,  satisfies:
    $$
p_{joint}(\chi, \sigma) =
\begin{cases}
\frac{1}{2} \pr(\sigma) & \text{if } \chi = \chi_{\sigma} \text{ or } \chi^r_{\sigma}, \\
0 & \text{otherwise}.
\end{cases}
$$
We say that a vertex $v$ belongs to the buyer set if $\chi(v)=B$ and $v$ belongs to the item set if $\chi(v)=I$.
\end{definition}

\begin{claim}\label{lem:2-coloring}
     The following property holds for the random coloring $ \chi$ defined in \cref{def:2-coloring}:
    \begin{enumerate}
        \item For all edge $(u,v)\in E$, $u,v$ get different colors when $(u,v)$ is included in $\M(\sigma)$.\\ 
        Formally, $\pr[\chi(u)\neq\chi(v)\;|\; (u,v)\in \M(\sigma)]=1$.
        \item For all edge $(u,u^*)\in M^*$, $u,u^*$ always get different colors.\\
         Formally, for all $(u,u^*)\in M^*$, $\pr[\chi(u)\neq\chi(u^*)]=1$.
    \end{enumerate}
    Further, the coloring of each vertex is uniform and independent of $\sigma$:
    \begin{enumerate}
        \setcounter{enumi}{2}
        \item For all vertex $v:\; \pr(\chi(v)=\B)=\pr(\chi(v)=\I)=\frac{1}{2}.$
        \item For all vertex $v,\;\text{permutation }\sigma:\;  \pr(\chi(v)=\B\;\wedge
        \;\sigma)=\pr(\chi(v)=\B)\cdot\pr(\sigma)=\frac{1}{2}\pr(\sigma)$.
    \end{enumerate}
\end{claim}
\begin{proof}
Properties (1), (2) are clear from the way we construct $\chi$. When $\sigma$ is realized, either $\chi=\chi_{\sigma}$ or $\chi=\chi_\sigma^r$, both of which are two colorings on $\M(\sigma)\cup M^*$. (3) comes from the following simple calculation: For any $v\in V$
$$\pr(\chi(v)=B)=\sum_{\sigma}\left[\mathds{1}_{\chi_{\sigma}(v)=\B}\cdot p_{joint}(\chi_{\sigma},\sigma)+\mathds{1}_{\chi^r_{\sigma}(v)=\B}\cdot p_{joint}(\chi^r_{\sigma},\sigma)\right]=\sum_{\sigma}\frac{1}{2}\pr(\sigma)=\frac{1}{2}.$$
Property (4) is immediate from fact (3) and the way we define our coloring.
\end{proof}
Properties $(3)$ and $(4)$, in particular, enable us to always assume that \( u \) takes a specific color for free when calculating the expected gain. This is illustrated by the following fact:
\begin{claim} \label{assumeuB}
    Let $\alpha\in[0,1]$ be a value. Let $\sigma$ be the bucketed random vector. Let $\ev_\sigma$ denote the expected value with respect to the distribution of $\sigma$. If for graph $G$, and all edges $(u,u^*)\in M^*$, 
    $$\alpha\leq \ev_{\sigma}[g(u)+g(u^*)\;|\; \chi(u)=\B]. $$
    Then for $G$ and all edges $(u,u^*)\in M^*$,
    $$ \alpha\leq \ev_{\sigma}[g(u)+g(u^*)]. $$
\end{claim}
\begin{proof}
     Assume $\alpha\leq \ev_{\sigma}[g(u)+g(u^*)\;|\; \chi(u)=\B]$ for $G$ and all pairs of vertices $(u,u^*)\in M^*$. In particular, the assumption holds for $G,(u,u^*)$ and $G,(u^*,u)$. So we have
     \begin{align*}
         \alpha\leq \ev_{\sigma}[g(u)+g(u^*)\;|\; \chi(u)=\B]\;\wedge\;\alpha\leq \ev_{\sigma}[g(u^*)+g(u)\;|\; \chi(u^*)=\B].
     \end{align*}
    By property $(4)$ of \cref{lem:2-coloring}, coloring on a single vertex does not affect the distribution, we have
    \begin{align*}
        \alpha\leq \frac{1}{2}\ev_{\sigma}[g(u)+g(u^*)\;|\; \chi(u)=\B]+\frac{1}{2}\ev_{\sigma}[g(u)+g(u^*)\;|\; \chi(u)=\I]= \ev_{\sigma}[g(u)+g(u^*)].
    \end{align*}
\end{proof}
Due to the above claim, we can always assume for free a fixed color for $u$.\textbf{ From now on, we will assume that $u$ always corresponds to a buyer, i.e. $\chi(u)=B$. By \cref{lem:2-coloring}, this implies that $u^*$ is always an item.}
\subsection{Price Function and Gain Sharing}
The random two-partitioning enables us to define a gain-sharing function assuming all edges included by \ranking{} have a buyer on one end and an item on the other end. We now define the set of price functions considered by our LP for the $k$-bucketed random permutation.
\begin{definition}[Price Function $\price$]\label{pricefun}
A \textbf{price function} is a function $\price:\{1,...,k\}\times\{1,...,k\}\to [0,1]$ that takes a (buyer, item) pair and outputs a price between $0$ and $1$ that satisfies the following monotonicity constraints:
    \begin{equation}\label{pricefunction}
        \begin{split}
             &\quad \price(i,j)\geq \price(i+1,j)\quad\forall i<k,\forall j,\\
        &\quad \price(i,j)\leq \price(i,j+1)\quad\forall i, \forall j<k.
        \end{split}
    \end{equation}
\end{definition}
Notice that the domain of $\price$ is restricted to $k$, the number of buckets. Thus, the price function will only depend on the bucket value $\x$ rather than $y$, the relative rank within buckets (which will be relevant to  the number of vertices in the graph). This allows our analysis to extend to arbitrarily large graphs using a fixed number of buckets $k$. 

Further, the price function depends on both buyers and items. When the first input takes the bucket value of a buyer and the second input takes the bucket value of an item, the monotonicity constraints naturally correspond to the intuition that for both buyers and items, smaller ranked counterparts are preferred over larger ranked ones. This is illustrated by \cref{monotonicitygainfunction}.

Now, with the price function defined, we can define the gain-sharing function $g: V \to [0, 1]$, where, as promised, $\sum_{u \in V} g(u) = |\M(\sigma)|$. Intuitively, for an edge $(u, v)$ matched in $\M(\sigma)$, where $u$ is the buyer and $v$ is the item, we assign $\price(\x_u, \x_v)$ amount of gain to $v$ and assign $1 - \price(\x_u, \x_v)$ amount of gain to $u$. This is a well-defined gain-sharing scheme by the way we defined our partitioning. Formally, it is defined by the following.
\begin{definition}
    Let $\sigma$ be a rank vector on $V$, and let $\M(\sigma)$ be the respective matching. Let $\chi$ be a realization of the random two-partitioning defined in \cref{def:2-coloring}. Let $\price$ 
    be a price function as defined in \cref{pricefun}. The gain sharing function $g:V\to [0,1]$ with respect to $\sigma$, $\chi,$ is defined as
$$\g_{\chi,\sigma}(u)=\begin{cases}
    \price(\x_v,\x_u),  &\text{ if }\chi(v)=\B,\; \chi(u)=\I,\; (u,v)\in \M(\sigma)\\
    1-\price(\x_u,\x_v), &\text{ if }\chi(u)=\B,\; \chi(v)=\I,\; (u,v)\in  \M(\sigma)\\
    0.&\text{ otherwise}
\end{cases}$$
\end{definition}
With $g$ defined on each instance of $\chi$ and $\sigma$, it admits a distribution induced by their joint distribution. The expected summation of $g$ equals the expected size of the matching. So we can use $\ev[g(u)+g(u^*)]$ to bound the approximation ratio for \ranking{} as in \cref{guustarisalowerbound}. This is illustrated by the following:
\begin{fact}
    Let $g$ be the random variable s.t. the probability of $g=g_{\chi,\sigma}$ is $=p_{joint}(\chi,\sigma)$, where $p_{joint}(\chi,\sigma)$ is defined in \cref{def:2-coloring}. We have
    $$\ev\left[\sum_{u\in V}g(u)\right]=\ev[|\M(\sigma)|].$$
\end{fact}
\begin{proof}
    This is simply because for all $\chi_{\sigma},\chi_{\sigma}^r,\sigma$ we have 
    $$\sum_{u\in V}g_{\chi_{\sigma},\sigma}(u)=\sum_{u\in V}g_{\chi_{\sigma}^r,\sigma}(u)=|\M(\sigma)|.$$
\end{proof}
The reason we defined the price function in \cref{pricefun} with two monotonicity constraints is illustrated by the following two monotonicity properties of the gain function $g$. These properties reflect the intuition that both buyers and items receive larger gains when matched with counterparts with smaller ranks.
\begin{fact}\label{monotonicitygainfunction}
    If $u$ is matched to some vertex $v$ with bucket value $\x_v \leq i$ in $\M(\sigma)$, where $u$ is a buyer and $v$ is an item, then the gain of $u$ satisfies:
    $$ g(u) \geq 1 - \price(\x_u, i).$$
    Similarly, if $v$ is matched to some vertex $u$ with bucket value $\x_u \leq i$ in $\M(\sigma)$, where $u$ is a buyer and $v$ is an item, then the gain of $v$ satisfies:
    $$ g(v) \geq \price(i, \x_v).$$
\end{fact}

\begin{proof}
    This follows directly from the monotonicity properties of the price function $\price$ as defined in \cref{pricefun}.
\end{proof}

\section{Lower Bounding Gains for Fixed $\sigma_{-u^*}$}\label{section6}
The overall plan we use to calculate $\mathbb{E}[g(u) + g(u^*)]$ is to consider rank vectors $\sigma_{-u^*}$ on the graph with $u^*$ removed first. We can add $u^*$ uniformly to each fixed $\sigma_{-u^*}$ and calculate the expected gain conditioning on $\sigma_{-u^*}$ being realized. Another round of expectation is further taken over the distribution of $\sigma_{-u^*}$ to calculate the final expected gain. Formally, we will be calculating:
$$\ev_{\sigma_{-u^*}}\left[\ev_{\x_{u^*}}\left[g(u)+g(u^*)\;|\;\sigma_{-u^*}\right]\right].$$
In this section, we focus on approximating the inner expected value 
$$\ev_{\x_{u^*}}[g(u)+g(u^*)\;|\;\sigma_{-u^*}],$$
by building a case-by-case lower bound for each fixed $(\sigma_{-u^*}, \x_{u^*})$ pair. We categorize $\sigma_{-u^*}$ into three different cases: $\classnomatch,\classnobackup:$ and $\classbackup$. We will define three functions: $\hclassnomatch,\hclassnobackup$ and $\hclassbackup$ to account for each of the three cases.  

Let $\sigma$ be a rank vector on $V$ and let $\sigma_{-u^*}$ be the induced vector with $u^*$ removed. When $\sigma_{-u^*}$ belongs one of the three cases $\classnomatch,\classnobackup:$ and $\classbackup$, the corresponding $h$ function will take the profile of $u$ in $\sigma_{-u^*}$ and $\x_{u^*}$ as input. $h$ will then output a lower bound for $g(u)+g(u^*)$ when $\sigma$ is realized.

After defining $\hclassnomatch$, $\hclassnobackup$, and $\hclassbackup$, we use the fact that the bucket value of $u^*$ is independent of $\sigma_{-u^*}$ to calculate the expected lower bound for $g(u) + g(u^*)$ conditioned on $\sigma_{-u^*}$. Specifically, this corresponds to the inner expectation value $\ev_{\x_{u^*}}[g(u) + g(u^*) \;|\; \sigma_{-u^*}]$.

We will use the notion of profiles for $u$ in rank vectors, with slight adjustments from \cref{def:profile}. In a bucketed rank vector, the profile of $u$ depends only on the bucket values, not the full rank values. We then classify each rank vector $\sigma_{-u^*}$ based on its respective profile of $u$.
\begin{definition}[Profile with respect to rank vectors]
     Let $\sigma_{-u^*}$ be a rank vector on $V \backslash \{u^*\}$. The profile of $u$ in $\sigma_{-u^*}$ is a triple $(\x_u, \x_v, \x_b)$ such that:
\begin{enumerate}
    \item $\x_u$ is the bucket value of $u$ in $\sigma_{-u^*}$.
    \item $\x_v$ is the bucket value of $v$ in $\sigma_{-u^*}$ if $u$ is matched to $v$ in $\M(\sigma_{-u^*})$; $\x_v = \bot$ otherwise.
    \item $\x_b$ is the bucket value of $b$ in $\sigma_{-u^*}$ if $b$ is the backup of $u$ in $\sigma_{-u^*}$; $\x_b = \bot$ otherwise.
\end{enumerate}
\end{definition}
\begin{definition}\label{def:Cs}
   Let $(\x_u,\x_v,\x_b)$ be the profile of $u$ in $\sigma_{-u^*}$. We classify  $\sigma_{-u^*}$ into one of the three cases $\classbackup,\classnobackup,\classnomatch$ based on the corresponding criteria:
    \begin{align*}
        &\classnomatch: \text{ if $u$ is not matched in $\M(\sigma_{-u^*})$; or equivalently, $\x_v,\x_b=\bot$,}\\
        &\classnobackup: \text{ if $u$ is matched in $\M(\sigma_{-u^*})$ but does not have a backup; or equivalently, $\x_v\neq\bot,\;\x_b= \bot$},\\
        &\classbackup: \text{ if $u$ is matched in $\M(\sigma_{-u^*})$ and has a backup; or equivalently, $\x_v,\x_b\neq \bot$}.
    \end{align*}
\end{definition}
The subscripts $\bot$, $s$, and $b$ mean: empty of a match ($\bot$), having a single choice when picking, thus does not have a backup (s), and having a backup ($b$).

We now define lower bound functions $\hclassnomatch$, $\hclassnobackup$, and $\hclassbackup$. Similar to the gain sharing function, these functions will only depend on the bucket values.
\subsection{Function $\hclassnomatch$}\label{sectionhbot}
We define $\hclassnomatch:\{1,...,k\}^2\to [0,2]$, the lower bound function for $\sigma_{-u^*}\in\classnomatch$ is the following:
\begin{definition}
    For any two bucket values $\x_u,\x_{u^*}\in\{1,...,k\}$,
\begin{equation}\label{defhclassnomatch}
    \hclassnomatch(\x_u,\x_{u^*})=\price(\x_{u},\x_{u^*}).
\end{equation}
\end{definition}

Assume $\sigma$ is a rank vector such that the induced matching $\sigma_{-u^*}$ is classified as $\classnomatch$. Then the profile of $u$ in $\sigma_{-u^*}$ is $(\x_u, \bot, \bot)$. Let $\M(\sigma)$ be the result of \ranking{} on $\sigma$ and let the bucket value of $u^*$ be $\x_{u^*}$. We have the following:
\begin{claim}\label{lem:hclassnomatch}
$\hclassnomatch(\x_u, \x_{u^*})$ lower bounds $\g(u) + \g(u^*)$ when $\sigma$ is realized.
\end{claim}
\begin{proof}
    Recall by \cref{assumeuB}, we assumed that $u$ is a buyer and $u^*$ is an item. By \cref{lem:nomatch}, $u^*$ is always matched to some vertex of rank at most $(\x_u, \y_u)$, in which case it receives a gain no less than $\price(\x_u, \x_{u^*})$ by \cref{monotonicitygainfunction}.
\end{proof}

\subsection{Function $\hclassnobackup$}\label{sectionhs}
We define $\hclassnobackup:\{1,...,k\}^3\to [0,2]$, the lower bound function for $\sigma_{-u^*}\in \classnobackup$  by the following:
\begin{definition}
    For any three bucket values $\x_u, \x_v, \x_{u^*}\in\{1,...,k\}$,
\begin{equation}\label{defhclassnobackup}
\hclassnobackup(\x_u,\x_v,\x_{u^*})=
\begin{cases} 
    \price(\x_u,\x_{u^*}), & \text{if } \x_u < \x_v, \x_{u^*} \leq \x_u, \\
    1 - \price(\x_u,\x_v) + \price(\x_u,\x_{u^*}), & \text{if } \x_u < \x_{u^*} < \x_v, \\
    1 - \price(\x_u,\x_v), & \text{if } \x_u < \x_v, \x_v \leq \x_{u^*}, \\
    \min\{\price(\x_v,\x_{u^*}), 1 - \price(\x_u,\x_v) + \price(\x_u,\x_{u^*})\}, & \text{if } \x_v \leq \x_u, \x_{u^*} < \x_v, \\
    \min\{\price(\x_v,\x_{u^*}), 1 - \price(\x_u,\x_v)\}, & \text{if } \x_v \leq \x_{u^*} \leq \x_u, \\
    1 - \price(\x_u,\x_v) & \text{if }\x_v \leq \x_u , \x_u < \x_{u^*}.
\end{cases}
\end{equation}
\end{definition}
Assume $\sigma$ is a rank vector such that the induced matching $\sigma_{-u^*}$ is classified as $\classnobackup$. Then the profile of $u$ in $\sigma_{-u^*}$ is $(\x_u, \x_v, \bot)$. Let $\M(\sigma)$ be the result of \ranking{} on $\sigma$ and let the bucket value of $u^*$ be $\x_{u^*}$. We have the following:
\begin{claim}\label{lem:hclassnobackup}
    $\hclassnobackup(\x_u,\x_v,\x_{u^*})$ lower bounds $\g(u)+\g(u^*)$ when $\sigma$ is realized.
\end{claim}
\begin{proof}
    Recall by \cref{assumeuB}, we assumed that $u$ is a buyer and $u^*$ is an item. Also recall the monotonicity property of gain function $g$ as in \cref{monotonicitygainfunction}.

\textbf{Case ( $\x_u < \x_v, \x_{u^*} \leq \x_u$):} $\x_{u^*}\leq \x_u<\x_v\Rightarrow \sigma(u^*)<\sigma(v)$. By \cref{lem:struct-property}, $u^*$ is matched to some vertex with rank at most $\sigma(u)$, hence $\g(u)+\g(u^*)\geq \g(u^*)\geq \price(\x_u,\x_{u^*})$.
    
\textbf{Case ($\x_u<\x_{u^*} < \x_v$):} $\x_u<\x_{u^*}<\x_v\Rightarrow \sigma(u) < \sigma(u^*)<\sigma(v)$. By \cref{lem:struct-property}, $u^*$ is matched to some vertex with rank at most $\sigma(u)$, and $u$ is matched to some vertex with rank at most $\sigma(v)$. So we have $\g(u^*)\geq \price(\x_u,\x_{u^*})$, and $\g(u)\geq 1-\price(\x_u,\x_v)$. Hence $\g(u)+\g(u^*)\geq 1 - \price(\x_u,\x_v) + \price(\x_u,\x_{u^*})$.
    
\textbf{Case ($\x_u<\x_v, \x_{u^*} \geq \x_v$):} $\x_u<\x_v\leq \x_{u^*}\Rightarrow \sigma(u)<\sigma(u^*)$. By \cref{lem:struct-property} $u$ is matched to some vertex with rank at most $\x_v$. So we have $\g(u)+\g(u^*)\geq \g(u)\geq 1-\price(\x_u,\x_v)$.
    
\textbf{Case ( $\x_u \geq \x_v, \x_{u^*} < \x_v$):} $\x_{u^*}<\x_v\Rightarrow \sigma(u^*)<\sigma(v)$. By \cref{lem:struct-property}, $u^*$ is matched to some vertex with rank at most $\sigma(u)$. If $u^*$ being matched does not make $u$ worse off\footnote{Worse off means $u$ gets matched to some vertex with rank $>\sigma(v)$}, then $\g(u)+\g(u^*)\geq 1-\price(\x_u,\x_v)+\price(\x_u,\x_{u^*})$. If $u^*$ being matched makes $u$ worse off, in which case $u^*$ is matched to someone with rank at most $\sigma(v)$. In this case $\g(u)+\g(u^*)\geq \price(\x_v,\x_{u^*})$. So we have $\g(u)+\g(u^*)\geq\min\{ \price(\x_v,\x_{u^*}), 1 - \price(\x_u,\x_v) + \price(\x_u,\x_{u^*})\}$.
    
\textbf{Case ( $\x_v\leq \x_{u^*} \leq \x_u$):} Any ordering of $\sigma(u),\sigma(v),\sigma(u^*)$ can happen, there are two cases, either $u$ is made worse off by $u^*$ or not. By \cref{lem:struct-property}, if $u$ is made worse off, then $u^*$ is matched to some vertex with rank at most $\sigma(v)$. In this case we have $\g(u)+\g(u^*)\geq g(u^*)\geq \price(\x_v,\x_{u^*})$. If $u$ is not made worse off, we have $\g(u)+\g(u^*)\geq \g(u)\geq 1-\price(\x_u,\x_v)$. So we have $\g(u)+\g(u^*)\geq\min\{\price(\x_v,\x_{u^*}), 1 - \price(\x_u,\x_v)\}$.
    
\textbf{Case ($\x_u\geq \x_v, \x_{u^*} > \x_u$):} $\x_u<\x_{u^*}\Rightarrow \sigma(u)<\sigma(u^*)$. By \cref{lem:struct-property}, $u$ is matched to some vertex with rank at most $\sigma(v)$. So we have $\g(u)+\g(u^*)\geq \g(u)\geq 1-\price(\x_u,\x_v)$.
\end{proof}
\subsection{Function $\hclassbackup$}\label{sectionhb}
We define $\hclassnobackup:\{1,...,k\}^4\to [0,2]$, the lower bound function for $\sigma_{-u^*}\in \classbackup$  by the following:
\begin{definition}
    For any four bucket values $\x_u, \x_v,\x_b, \x_{u^*}\in\{1,...,k\}$,
    \begin{equation}\label{defhclassbackup}
    \hclassbackup(\x_u,\x_v,\x_b,\x_{u^*})=
\begin{cases} 
    1 - \price(\x_u,\x_b) + \price(\x_u,\x_{u^*}) & \text{if } \x_u < \x_v, \x_{u^*} \leq \x_u, \\
    1 - \price(\x_u,\x_v) + \price(\x_u,\x_{u^*}) & \text{if } \x_u < \x_{u^*} < \x_v, \\
    1 - \price(\x_u,\x_v), & \text{if } \x_u < \x_v,\x_v \leq \x_{u^*}, \\
    \min\{1 - \price(\x_u,\x_b) + \price(\x_v,\x_{u^*}), & \text{if } \x_v \leq \x_u, \x_{u^*} < \x_v,\\
    \quad \quad1 - \price(\x_u,\x_v) + \price(\x_u,\x_{u^*})\}  \\
    \min\{1 - \price(\x_u,\x_b) + \price(\x_v,\x_{u^*}), & \text{if } \x_v \leq \x_{u^*} \leq \x_u, \\
    \quad \quad1 - \price(\x_u,\x_v)\} \\
    1 - \price(\x_u,\x_v) & \text{if }\x_v\leq \x_u, \x_u < \x_{u^*}.
\end{cases}
\end{equation}
\end{definition}
Assume $\sigma$ is a rank vector such that the induced matching $\sigma_{-u^*}$ is classified as $\classbackup$ equivalently, the profile of $u$ in $\sigma_{-u^*}$ is $(\x_u, \x_v, \x_b)$. Let $\M(\sigma)$ be the result of \ranking{} on $\sigma$ and let the bucket value of $u^*$ be $\x_{u^*}$. We have the following:
\begin{claim}\label{lem:hclassbackup}
     $\hclassbackup(\x_u,\x_v,\x_b,\x_{u^*})$ lower bounds $\g(u)+\g(u^*)$ when $\sigma$ is realized.
\end{claim}
\begin{proof}
Recall by \cref{assumeuB}, we assumed that $u$ is a buyer and $u^*$ is an item. Also, recall the monotonicity property of the gain function $g$ as in \cref{monotonicitygainfunction}.

\textbf{Case ( $\x_u < \x_v, \x_{u^*} \leq \x_u$):} $\x_{u^*} \leq \x_u < \x_v \Rightarrow \sigma(u^*) < \sigma(v)$. By \cref{lem:struct-property}, $u^*$ is matched to some vertex with rank at most $\sigma(u)$, which implies $\g(u^*) \geq \price(\x_u, \x_{u^*})$. By \cref{lem:backup}, $u$ is matched to some vertex with rank at most $\sigma(b)$. This implies $\g(u) \geq 1 - \price(\x_u, \x_b)$. Therefore, we have:
$\g(u) + \g(u^*) \geq 1 - \price(\x_u, \x_b) + \price(\x_u, \x_{u^*}).$

\textbf{Case ($\x_u<\x_{u^*} < \x_v$):} $\x_u < \x_{u^*} < \x_v \Rightarrow \sigma(u) < \sigma(u^*) < \sigma(v)$. Hence, by \cref{lem:struct-property}, $u^*$ is matched to some vertex with rank at most $\sigma(u)$, and $u$ is matched to some vertex with rank at most $\sigma(v)$. This implies $\g(u^*) \geq \price(\x_u, \x_{u^*})$ and $\g(u) \geq 1 - \price(\x_u, \x_v)$. Therefore, we have:$\g(u) + \g(u^*) \geq 1 - \price(\x_u, \x_v) + \price(\x_u, \x_{u^*}).$

\textbf{Case ($\x_u<\x_v, \x_{u^*} \geq \x_v$):}$\x_u < \x_v \leq \x_{u^*} \Rightarrow \sigma(u) < \sigma(u^*)$. By \cref{lem:struct-property}, $u$ is matched to some vertex with rank at most $\sigma(v)$, and we have $\g(u) + \g(u^*) \geq \g(u) \geq 1 - \price(\x_u, \x_v)$.

\textbf{Case ( $\x_u \geq \x_v, \x_{u^*} < \x_v$):} $\x_{u^*}<\x_v\Rightarrow \sigma(u^*)<\sigma(v)$. By \cref{lem:struct-property}, $u^*$ is matched to some vertex with rank at most $\sigma(u)$. Either $u^*$ being matched does not make $u$ worse off, in which case $\g(u) + \g(u^*) \geq 1 - \price(\x_u, \x_v) + \price(\x_u, \x_{u^*})$; or $u^*$ being matched made $u$ worse off, in which case $u$ is matched to some vertex with rank at most $\sigma(b)$ (by \cref{lem:backup}) and $u^*$ is matched to some vertex with rank at most $\sigma(v)$, in this case $\g(u) + \g(u^*) \geq 1 - \price(\x_u, \x_b) + \price(\x_v, \x_{u^*})$. So we have 
$\g(u) + \g(u^*) \geq \min \{ 1 - \price(\x_u, \x_b) + \price(\x_v, \x_{u^*}), 1 - \price(\x_u, \x_v) + \price(\x_u, \x_{u^*}) \}.$

\textbf{Case ( $\x_v\leq \x_{u^*} \leq \x_u$):} Essentially any ordering of $\sigma(u), \sigma(v), \sigma(u^*)$ can happen. Either $u$ is made worse off by $u^*$ or not. If $u$ is made worse off, by \cref{lem:backup} it will still be matched to some vertex with rank at most $\sigma(b)$, and in this case $u^*$ is matched to some vertex with rank at most $\sigma(v)$ by \cref{lem:struct-property}. Thus we have $\g(u) + \g(u^*) \geq 1 - \price(\x_u, \x_b) + \price(\x_v, \x_{u^*}).$ If $u$ is not made worse off, $\g(u) + \g(u^*) \geq \g(u) \geq 1 - \price(\x_u, \x_v)$. So we have $\g(u) + \g(u^*) \geq \min \left\{ 1 - \price(\x_u, \x_b) + \price(\x_v, \x_{u^*}), 1 - \price(\x_u, \x_v) \right\}.$

\textbf{Case ($\x_u\geq \x_v, \x_{u^*} > \x_u$):} $\x_u<\x_{u^*}\Rightarrow \sigma(u)<\sigma(u^*)$. By \cref{lem:struct-property} $u$ is matched to some vertex with rank at most $\sigma(v)$, and we have $\g(u)+\g(u^*)\geq \g(u)\geq 1-\price(\x_u,\x_v)$.
\end{proof}
\subsection{Expected Gain for Fix $\sigma_{-u^*}$}
Functions $\hclassnomatch, \hclassnobackup, \hclassbackup$ provide lower bounds for each fixed $(\sigma_{-u^*}, \x_{u^*})$ pairs. Now we can bound the expected gain conditioning on each fixed permutation $\sigma_{-u^*}$ by taking expectation over $\x_{u^*}$ which uniformly takes value in $\{1,...,k\}$. 

We will use functions $\Hclassnomatch, \Hclassnobackup, \Hclassbackup$ for the lower bounds of expected gain conditioned on fixed $\sigma_{-u^*} \in \classnomatch, \classnobackup, \classbackup$. $\Hclassbackup, \Hclassnobackup$ and $\Hclassnomatch$ are defined as follows:
\begin{definition}\label{def:Hfunctions}
    For bucket values $\x_u,\x_v,\x_b\in\{1,...,k\}$, $\Hclassnomatch, \Hclassnobackup, \Hclassbackup$ are defined as:
    \begin{align*}
        &\Hclassnomatch(\x_u) = \sum_{i=1}^\numpieces \frac{1}{\numpieces} \hclassnomatch(\x_u,i),\\
        &\Hclassnobackup(\x_u,\x_v) = \sum_{i=1}^\numpieces \frac{1}{\numpieces} \hclassnobackup(\x_u,\x_v,i)\\
        &\Hclassbackup(\x_u,\x_v,\x_b) = \sum_{i=1}^\numpieces \frac{1}{\numpieces} \hclassbackup(\x_u,\x_v,\x_b,i)\\
    \end{align*}
\end{definition}
\begin{claim}\label{lem:Hfunction}
     Let $\sigma_{-u^*}$ be a rank vector over $V \backslash \{u^*\}$, and let $(\x_u, \x_v, \x_b)$ be the profile of $u$ in $\sigma_{u^*}$. $\x_v,\x_b$ could be $\bot$ depending on which class $\sigma_{-u^*}$ belongs to. We have the following:
     \begin{itemize}
         \item If $\sigma_{-u^*} \in \classnomatch$, then $\Hclassnomatch(\x_u) \leq \ev[g(u) + g(u^*) \;|\; \sigma_{-u^*}].$
         \item If $\sigma_{-u^*} \in \classnobackup$, then $\Hclassnobackup(\x_u, \x_v) \leq \ev[g(u) + g(u^*) \;|\; \sigma_{-u^*}].$
         \item If $\sigma_{-u^*} \in \classbackup$, then $\Hclassbackup(\x_u, \x_v, \x_b) \leq \ev[g(u) + g(u^*) \;|\; \sigma_{-u^*}].$
     \end{itemize}
\end{claim}
\begin{proof}
    By the way we defined bucketed random permutation in \cref{def:BucketedRPG}, the bucket value $u^*$ takes is independent from $\sigma_{-u^*}$. Let $\sigma$ be the random rank vector over $V$ and let $\sigma_{-u^*}$ be the random induced rank vector over $V\backslash\{u^*\}$ with $u^*$ removed. For each fixed rank vector $\sigma'$ on $V\backslash\{u^*\}$, each bucket value $i\in \{1,...,k\}$, we have the following:
    $$\pr(\x_{u^*}=i\;\wedge\;\sigma_{-u^*}=\sigma')=\frac{1}{k}\pr(\sigma')$$
    This fact combined with claim \cref{lem:hclassnomatch}, \cref{lem:hclassnobackup} and \cref{lem:hclassbackup} entails the claim.
\end{proof}

\section{Expectation over $\sigma_{-u^*}$ is Obtained by Uniform Profiles}\label{section8} 
In the previous section, we derived lower bounds for the expected gain conditioning on each specific $\sigma_{-u^*}$ being realized. The next step is to take the expectation over the distribution of $\sigma_{-u^*}$. To do this, we will define equivalence classes which further partition $\classnomatch,\classnobackup,$ and $\classbackup$ into smaller pieces of form $EC_\bot^i,EC_s^i$ and $EC_b^i$ s.t. 
$$\bigcup EC_\bot^i=\classnomatch, \quad \bigcup EC_s^i=\classnobackup,\quad \bigcup EC_b^i=\classbackup$$

 It will turn out that within each equivalence class, the profiles of $u$ share the same fixed $\x_u$, $\x_b$ bucket values. The distribution of $\x_v$, on the other hand, is almost a uniform distribution within each equivalence class. We will show that we can bound the expected value assuming a uniformly distributed $\x_v$. By considering the worst-case uniform distributions, we provide lower bounds for $EC_\bot^i,EC_s^i$ and $EC_b^i$ which then translates to bounds for $\classnomatch,\classnobackup,$ and $\classbackup$ by a simple law of total expectation argument. We then take the minimum of the three lower bounds and apply another round of law of total expectation argument to give a lower bound for any distribution on $\sigma_{-u^*}$ as a combination of distributions in $\classnomatch,\classnobackup,$ and $\classbackup$.
 
 As we will only focus on rank vectors $\sigma_{-u^*}$ on $V\backslash\{u^*\}$ in this section instead if considering rank vectors $\sigma$ over $V$, \textbf{we will abuse notation and use $\sigma$ instead of $\sigma_{-u^*}$ to represent rank vectors on $V\backslash\{u^*\}$ in the rest of this section.}

Intuitively, the equivalence classes are defined as follows:
\begin{itemize}
    \item Each $EC_\bot^i$ consists of a single rank vector $\sigma\in \classnomatch$.
    \item Each $EC_s^i$ (respectively, $EC_b^i$) consists of a set of rank vectors $\sigma \in \classnobackup$ (respectively, $\classbackup$). For all rank vectors in the same equivalence class, $u$ is matched to the same vertex $v$, furthermore, the rank vectors only differ at the rank of $v$. In other words, when $v$ is excluded from each of them, the induced rank vectors are the same.
\end{itemize}
Recall from \cref{descrete-monotonicity-nobackup} that we proved that increasing the rank of $v$ (the match of $u$) does not affect the matching. Also, recall from \cref{descrete-monotonicity-backup} that increasing the rank of $v$ does not affect the matching if, after the increase, $v$ is still placed before $b$ (the backup of $u$). These two properties entail the following structures for the equivalence classes $EC_s^i$ and $EC_b^i$:
\begin{itemize}
    \item $EC_s^i$: For each $EC_s^i$, there exists a generating rank vector $\sigma\in \classnobackup$. Denote $v$ as the match $u$ in $\M(\sigma)$. We have:
$$EC_s^i=\{\sigma_v^{(x,y)}\;|\; (x,y)\geq \sigma(v) \}$$
\item $EC_b^i$: For each $EC_b^i$, there exists a generating rank vector $\sigma\in \classbackup$. Denote $v$ as the match of $u$ in $\M(\sigma)$, and denote $b$ as the backup of $u$ in $\sigma$. We have:
$$EC_b^i=\{\sigma_v^{(x,y)}\;|\; \sigma(v)\leq (x,y)<\sigma(b)\}$$
\end{itemize}
With these structures, it holds that the distribution of the profile of $u$ within each equivalence class has $\x_u$ and $\x_b$ fixed. On the other hand, $\x_v$ is almost uniformly distributed except at the start and end buckets, namely $\x_0$ and $(\x_b \text{ or }\numpieces)$. At the start and end buckets, the $y$ values are also related. We cannot bound the distribution taking $y$ values into consideration as they depend on the size of the graph. We will prove that uniform distribution on $\x_v$ suffices to bound all such distributions involving $\y$ values. With this fact, we will construct lower bounds for each equivalence class assuming a uniform distribution structure on the profile of $u$.
\subsection{Definition of Equivalence Class}
We now formally define the equivalence classes and prove the structural properties of each equivalence class.
\begin{definition}\label{def:EC}
     Let $\sigma, \sigma'$ be rank vectors over $V \backslash \{u^*\}$. Let $v, v'$ be the match of $u$ in $\M(\sigma), \M(\sigma')$ respectively, if they exist. We say $\sigma, \sigma'$ are related, denoted as $\sigma \sim \sigma'$, if they belong to the same type of rank vector (\cref{def:Cs}) and the only difference appears at the rank of $u$'s match. Formally, $\sigma \sim \sigma'$ if one of the three cases holds:
    \begin{itemize}
    \item $(EC_\bot)$: $\sigma, \sigma' \in \classnomatch$ and $\sigma = \sigma'$,
    \item $(EC_s)$: $\sigma, \sigma' \in \classnobackup$ and $\sigma_{-v} = \sigma'_{-v'}$,
    \item $(EC_b)$: $\sigma, \sigma' \in \classbackup$ and $\sigma_{-v} = \sigma'_{-v'}$.
  \end{itemize}
    The $\sim$ relation defined above is an equivalence relation and can be easily verified. 
\end{definition}
 We will denote equivalence classes as $EC_\bot^i$, $EC_s^i$ and $EC_b^i$ indicating which of the three criteria is satisfied for rank vectors within the corresponding equivalence class. For example, any two vectors $\sigma,\sigma'\in EC_b^i$ will satisfy the $(EC_b)$ criterion.
\begin{claim}\label{sameexceptv}
Let $EC_\bot^i, EC_s^i$ and $EC_b^i$ be specific equivalence classes defined by $\sim$ partitioning $\classnomatch,\classnobackup$ and $\classbackup$ respectively. The following structural property holds for them, respectively. 
\begin{itemize}
    \item ($EC_\bot^i$): There exists a $\sigma\in \classnomatch$ s.t. $EC_\bot^i=\{\sigma\}$.
    \item ($EC_s^i$): There exists a $\sigma\in \classnobackup$ s.t. denoting $v$ as the match of $u$ in $\M(\sigma)$, we have:
    $$EC_s^i=\{\sigma_v^{(x,y)}\;|\; \forall (x,y)\; s.t.\; \sigma(v)\leq (x,y)\; \wedge\; \text{$\sigma_v^{(x,y)}$ is a valid rank vector}\}$$
    \item ($EC_b^i$): There exists $\sigma\in \classbackup$ s.t. denoting $v$ as the match of $u$ in $\M(\sigma)$, $b$ as the backup of $u$ in $\sigma$, we have:
    $$EC_b^i=\{\sigma_v^{(x,y)}\;|\; \forall (x,y)\; s.t.\; \sigma(v)\leq (x,y)<\sigma(b)\;\wedge\; \text{$\sigma_v^{(x,y)}$ is a valid rank vector}\}$$
We will denote this specific $\sigma$ as $\sigma_0$, the generating rank vector. 
\end{itemize}
\end{claim}
\begin{proof}
    The proof is essentially due to \cref{descrete-monotonicity-nobackup} and \cref{descrete-monotonicity-backup}, which allow us to increase the rank of $v$ while maintaining $v$ as the match of $u$. By the definition of equivalence class, all rank vectors in the same equivalence class share the same match $v$, and there exists a minimum rank to which $v$ can be moved such that $u$ remains matched with $v$. We can place $v$ arbitrarily between this minimum rank and the end rank ($(k,k)$ in the case of $EC^i_s$ and $\sigma(b)$ in the case of $EC_b^i$), while ensuring that the permuted rank vector still belongs to the same equivalence class. A detailed proof can be found in \nameref{Appendix1}.
\end{proof}

With the equivalence classes defined, we will provide bounds for each of the equivalence classes $EC_\bot^i,EC_s^i$ and $EC_b^i$ and the respective entailed lower bounds for $\classnomatch,\classnobackup$ and $\classbackup$. 

\subsection{Lower Bounds For $EC_\bot^i$ and $\classnomatch$}\label{sectionexpecthbot}
We start with the easiest case, $EC_\bot^i$. Fix an arbitrary equivalence class $EC_\bot^i$. By \cref{sameexceptv}, it consists of a singleton rank vector $\sigma_0$. Let $(\x_u,\bot,\bot)$ be the profile of $u$ in $\sigma_0$. The following is a lower bound for expected gain conditioning on $\sigma\in EC_\bot^i$. 
\begin{claim}\label{lowerboundingecboti}
    Let $\sigma$ be the random rank vector over $V\backslash\{u^*\}$. The expected gain of $g(u)+g(u^*)$ conditioning on $\sigma\in EC_\bot^i$ is lower bounded by:
$$
    \Hclassnomatch(\x_u) = \sum_{\x_{u^*}=1}^\numpieces \frac{1}{\numpieces}\hclassnomatch(\x_u,\x_{u^*}).
$$
\end{claim}
\begin{proof}
    As conditioning on $\sigma\in EC_\bot^i$ is conditioning on $\sigma=\sigma_0$. By \cref{lem:Hfunction}, $\Hclassnomatch(\x_u)$ is a lower bound for the expected value of $g(u)+g(u^*)$.
\end{proof}
 Because any $\sigma\in \classnomatch$ belongs to some equivalence class of type $EC_\bot^i$, a lower bound for $\ev[g(u)+g(u^*)]$ conditioning on $\sigma\in \classnomatch$ and $\x_u=i$ can be obtained by the following.
 \begin{claim}\label{claim7.1}
     Let $\sigma$ be the bucketed random rank vector over $V\backslash\{u^*\}$, and let $\x_u$ the the random bucket value for $i$ in $\sigma$. Let $i$ be a bucket value in $\{1,...,k\}$. We have that the expected value of $g(u)+g(u^*)$ conditioning on $\sigma\in \classnomatch$ and $\x_u=i$, that is, the value
     $$\ev [g(u)+g(u^*)\;|\; \sigma\in \classnomatch\;\wedge\;\x_u=i],$$
     is lower bounded by:
     \begin{equation}\label{expextedhbot}
     \sum_{\x_{u^*}=1}^\numpieces \frac{1}{\numpieces}\hclassnomatch(i,\x_{u^*}).
     \end{equation}
 \end{claim}
 \begin{proof}
By the definition of $\sim$ and \cref{sameexceptv}, there exists a list of equivalence classes $EC_\bot^{i_1}, EC_\bot^{i_2}, \ldots, EC_\bot^{i_l}$ partitioning the set of rank vectors $\sigma \in \classnomatch$ where bucket value of $u$ is $c$. Let 
$$\alpha_i = \sum_{\x_{u^*}=1}^\numpieces \frac{1}{\numpieces}\hclassnomatch(i, \x_{u^*}).$$
We have by \cref{lowerboundingecboti}:
$$
\alpha_i \leq \ev[g(u) + g(u^*) \;|\; \sigma \in EC_\bot^{i_j}] \quad \text{for all equivalence classes } EC_\bot^{i_j}.
$$
Thus:
\begin{align*}
    \alpha_i &\leq \sum_{j=1}^l \ev[g(u) + g(u^*) \;|\; \sigma \in EC_\bot^{i_j}] \cdot \pr(\sigma \in EC_\bot^{i_j} \;|\; \sigma \in \classnomatch \;\wedge\; \x_u = i), \\
    &= \ev[g(u) + g(u^*) \;|\; \sigma \in \classnomatch \;\wedge\; \x_u = i],
\end{align*}
by the law of total expectation.
 \end{proof}
\subsection{Lower Bounds For $EC_s^i$ and $\classnobackup$}\label{sectionexpecths}
Fix an arbitrary $EC_s^i$ set of rank vectors, and let $\sigma_0$ be the generating rank vector in $EC_s^i$ by \cref{sameexceptv}. Let $(\x_u,\x_0,\bot)$ be the profile of $u$ in $\sigma_0$. Let $v$ be the match of $u$ in $\M(\sigma_0)$. The following is a lower bound for expected gain conditioning on $EC_s^i$.
\begin{claim}
     Let $\sigma$ be the bucketed random rank vector over $V\backslash\{u^*\}$. The expected gain of $g(u)+g(u^*)$ conditioning on $\sigma\in EC_s^i$ is lower bounded by:
$$
    \min_{c\in\{1,...,\numpieces\}}\left\{\frac{1}{\numpieces-c+1}\sum_{\x=c}^{\numpieces}\sum_{\x_{u^*=1}}^{\numpieces}\frac{1}{\numpieces}\hclassnobackup(\x_u,\x,\x_{u^*})\right\}.
$$
\end{claim}
\begin{proof}
    Let $\sigma_0(v)=(\x_0,\y_0)$. By \cref{sameexceptv}, any rank vector in $EC_s^i$ is generated by moving $v$ in $\sigma_0$ to some rank $(x_v,y_v)\geq(\x_0,\y_0)$. 
    
    The probability of $\x_v=x_0$ conditioning on $\sigma\in EC_s^i$ is the following, due to independence of the bucket value of $v$, and the fact that $\y_v$ needs to be a number $\geq\y_0$:
    $$\pr(\x_v=\x_0\;|\; \sigma\in EC_s^i)=\frac{\pr(\y_v\geq\y_0\;|\; \x_v=\x_0)\;\pr(\x_v=\x_0)\;}{\pr((x_v,y_v)\geq (\x_0,\y_0))}$$
    The probability of $\x_v=x$ for each $x>\x_0$ conditioning on $\sigma\in EC_s^i$ is the following, due to independence of the bucket value of $v:$
     $$\pr(\x_v=\x\;|\; \sigma\in EC_s^i)=\frac{\pr(\x_v=\x)}{\pr((x_v,y_v)\geq (\x_0,y_0))}$$
    Let  $p_0=\pr(\x_v=\x_0)\pr(\y_v\geq\y_0\;|\; \x_v=\x_0\;\wedge\; \sigma\in EC_s^i)$, we know that $0\leq p_0\leq \frac{1}{k}$. We also know that for each $x$, $\pr(\x_v=x)=\frac{1}{k}$. So the conditional probabilities are simplified to the following:
\begin{align*}
    &\pr(\x_v=\x_0\;|\; \sigma\in EC_s^i)=\frac{p_0}{p_0+\frac{\numpieces-\x_0}{\numpieces}},\\
     &\pr(\x_v=\x\;|\; \sigma\in EC_s^i)=\frac{\frac{1}{\numpieces}}{p_0+\frac{\numpieces-\x_0}{\numpieces}}\quad\quad\forall x>\x_0.
\end{align*}
 With this distribution of $\x_v$, we have a lower bound for $g(u) + g(u^*)$ conditioning on $\sigma\in EC_s^i$ by
\begin{align*}
    &\ev[g(u)+g(u^*)\;|\; \sigma\in EC_s^i],\\
    \geq&\pr(\x_v=\x_0\;|\; \sigma\in EC_s^i)H_s(\x_u,\x_0)+\sum_{\x=\x_0+1}^{\numpieces}\pr(\x_v=\x\;|\; \sigma\in EC_s^i)H_s(\x_u,x),\\
    =&\mathbf{a_1}\cdot H_s(\x_u,\x_0)+\mathbf{a_2}\cdot\sum_{\x=\x_0+1}^{\numpieces}H_s(\x_u,x),
\end{align*}
where $H_s$ is the function defined in \cref{def:Hfunctions} and denote $\mathbf{a_1}=\frac{p_0}{p_0+\frac{\numpieces-\x_0}{\numpieces}}$, $\mathbf{a_2}=\frac{\frac{1}{\numpieces}}{p_0+\frac{\numpieces-\x_0}{\numpieces}}$. When $\x_u, \x_0$ are fixed constants, which is true as it is the profile of $u$ in $\sigma_0$, the two terms multiplied by $\mathbf{a_1}$ and $\mathbf{a_2}$ are constants. Thus, the lower bound is simply an affine combination of two constants, and the minimum is achieved in one of the following two cases:
\begin{align*}
    \frac{1}{\numpieces-x_0+1}\sum_{\x=\x_0}^{\numpieces}H_s(\x_u,\x)&=\frac{1}{\numpieces-x_0+1}\sum_{\x=\x_0}^{\numpieces}\sum_{\x_{u^*=1}}^{\numpieces}\frac{1}{\numpieces}\hclassnobackup(\x_u,\x,\x_{u^*}),&&\text{(when $\mathbf{a_1}$ is maximized)},\\
     \frac{1}{\numpieces-x_0}\sum_{\x=\x_0+1}^{\numpieces}H_s(\x_u,\x)&=\frac{1}{\numpieces-x_0}\sum_{\x=\x_0+1}^{\numpieces}\sum_{\x_{u^*=1}}^{\numpieces}\frac{1}{\numpieces}\hclassnobackup(\x_u,\x,\x_{u^*}).&&\text{(when $\mathbf{a_2}$ is maximized)}.
\end{align*}
Now we can derive a lower bound for the expected gain for $\sigma\in EC_s^i$, where $\x_u$ is fixed and $\x_0$ is chosen to be adversarial. The lower bound is the following, accounting for both cases above: 
$$
    \min_{c\in\{1,...,\numpieces\}}\{\frac{1}{\numpieces-c+1}\sum_{\x=c}^{\numpieces}\sum_{\x_{u^*=1}}^{\numpieces}\frac{1}{\numpieces}\hclassnobackup(\x_u,\x,\x_{u^*})\}.
$$
This is technically a lower bound stronger than needed, as we are choosing $\x_0$ to be adversarial. We are doing this because in the next step, we need to consider adversarial $\x_0$ but we are safe to assume a fixed $\x_u$ value.
\end{proof}
Similar to the $\classnomatch$ case, any $\sigma\in \classnobackup$ belongs to some equivalence class of type $EC_s^i$, a lower bound for $\ev[g(u)+g(u^*)\;|\; \sigma\in\classnobackup\;\wedge\;\x_u=i]$ can be obtained by setting $\x_u=i$.
\begin{claim}\label{claim7.2}
     Let $\sigma$ be the bucketed random rank vector over $V\backslash\{u^*\}$, and let $j$ be a bucket value in $\{1,...,k\}$. We have that the expected value of $g(u)+g(u^*)$ conditioning on $\sigma\in \classnobackup$ and $\x_u=i$, i.e. the value
     $$\ev[g(u)+g(u^*)\;|\; \sigma\in \classnobackup\;\wedge\;\x_u=i],$$
     is lower bounded by:
     \begin{equation}\label{expectedhs}
     \min_{c\in\{1,...,\numpieces\}}\left\{\frac{1}{\numpieces-c+1}\sum_{\x=c}^{\numpieces}\sum_{\x_{u^*=1}}^{\numpieces}\frac{1}{\numpieces}\hclassnobackup(i,\x,\x_{u^*})\right \}.
     \end{equation}
 \end{claim}
 \begin{proof}
 The proof is similar to the case $\classnomatch$. See detailed proof in \nameref{Appendix1}.
 \end{proof}
\subsection{Lower Bounds For $EC_b^i$ and $\classbackup$}\label{sectionexpecthb}
Fix an arbitrary $EC_b^i$ set of rank vectors, and let $\sigma_0$ be the generating rank vector in $EC_s^i$ by \cref{sameexceptv}. Let $(\x_u,\x_0,\x_b)$ be the profile of $u$ in $\sigma_0$. Let $v$ be the match of $u$ in $\M(\sigma_0)$. The following is a lower bound for the expected gain conditioning on $\sigma\in EC_b^i$.
\begin{claim}
     Let $\sigma$ be the bucketed random rank vector over $V\backslash\{u^*\}$. The expected gain of $g(u)+g(u^*)$ conditioning on $\sigma\in EC_b^i$ is lower bounded by:
$$
    \min_{c\leq d}\left\{\frac{1}{d-c+1}\sum_{\x=c}^{d}\sum_{\x_{u^*=1}}^{\numpieces}\frac{1}{\numpieces}\hclassbackup(\x_u,\x,d+1,\x_{u^*})\right\}.
$$
\end{claim}

\begin{proof}
Let $\sigma_0(v)=(\x_0,y_0)$ and let $\sigma_0(b)=(\x_b,\y_b)$. By \cref{sameexceptv}, any rank vector in $EC_b^i$ is generated by moving $v$ in $\sigma_0$ to some rank $(x_0,y_0)\leq (x_v,y_v)< (\x_b,\y_b)$. 

The probability of $\x_v=x_0$ conditioning on $\sigma\in EC_b^i$ is the following, due to independence of the bucket value of $v$, and the fact that $\y_v$ needs to be a number $\geq\y_0$:
    $$\pr(\x_v=\x_0\;|\; \sigma\in EC_b^i)=\frac{\pr(\y_v\geq\y_0\;|\; \x_v=\x_0)\;\pr(\x_v=\x_0)}{\pr((\x_0,y_0)\leq(\x_v,y_v)< (\x_b,y_b))}$$
    The probability of $\x_v=x$ for each $\x_0<x<\x_b$ conditioning on $\sigma\in EC_s^i$ is the following, due to independence of the bucket value of $v:$
     $$\pr(\x_v=\x\;|\; \sigma\in EC_s^i)=\frac{\pr(\x_v=\x)}{\pr((\x_0,y_0)\leq(\x_v,y_v)< (\x_b,y_b))}$$
    The probability of $\x_v=\x_b$ conditioning on $\sigma\in EC_b^i$ is the following, due to independence of the bucket value of $v$, and the fact that $\y_v$ needs to be a number $<\y_b$:
    $$\pr(\x_v=\x_b\;|\; \sigma\in EC_b^i)=\frac{\pr(\y_v<\y_b\;|\; \x_v=\x_b)\;\pr(\x_v=\x_b)}{\pr((\x_0,y_0)\leq(\x_v,y_v)< (\x_b,y_b))}$$
    
    Similar to the proof for $EC_s^i$ we know that there exists $p_0,p_1\in[0,\frac{1}{k}]$ s.t. the conditional probabilities can be simplified to the following:
\begin{align*}
    &\pr(\x_v=\x_0\;|\;\sigma\in EC_b^i)=\frac{p_0}{p_0+\frac{\x_b-1-\x_0}{\numpieces}+p_1},\\
     &\pr(\x_v=\x\;\;\;|\;\sigma\in EC_b^i)=\frac{\frac{1}{\numpieces}}{p_0+\frac{\x_b-1-\x_0}
     {\numpieces}+p_1},\quad\quad\forall x\text{ s.t. } x_0< x< \x_b,\\
     &\pr(\x_v=\x_b\;\;|\;\sigma\in EC_b^i)=\frac{p_1}{p_0+\frac{\x_b-1-\x_0}{\numpieces}+p_1}.
\end{align*}
For simplicity, we use the notation $\pr_{\x_v}(x)$ to represent $\pr(\x_v = x \;|\; \sigma \in EC_b^i)$. With this distribution of $\x_v$, we obtain a lower bound for $g(u) + g(u^*)$ conditioned on $\sigma\in EC_b^i$ as follows:
\begin{align*}
    &\ev[g(u)+g(u^*)\;|\;\sigma\in EC_b^i],\\
    \geq&\pr_{\x_v}(\x_0)\Hclassbackup(\x_u,\x_0,\x_b)
    +\sum_{\x=\x_0+1}^{\x_b-1}\pr_{\x_v}(x)\Hclassbackup(\x_u,\x,\x_b)
    +\pr_{\x_v}(\x_b)\Hclassbackup(\x_u,\x_b,\x_b),\\
    =&\mathbf{a_1}\cdot \Hclassbackup(\x_u,\x_0,\x_b)
    +\mathbf{a_2}\cdot \sum_{\x=\x_0+1}^{\x_b-1}\Hclassbackup(\x_u,x,\x_b)
    +\mathbf{a_3}\cdot \Hclassbackup(\x_u,\x_b,\x_b),
\end{align*}
where $H_b$ is defined in \cref{def:Hfunctions}, and
$\mathbf{a_1} = \frac{p_0}{p_0 + \frac{\x_b - 1 - \x_0}{\numpieces} + p_1},\;
\mathbf{a_2} = \frac{\frac{1}{\numpieces}}{p_0 + \frac{\x_b - 1 - \x_0}{\numpieces} + p_1}\;
\mathbf{a_3} = \frac{p_1}{p_0 + \frac{\x_b - 1 - \x_0}{\numpieces} + p_1}.$
When $\x_u, \x_0, \x_b$ are fixed constants, which is true as it is the profile of $u$ with respect to $\sigma_0$, the three terms multiplied by $\mathbf{a_1}, \mathbf{a_2}, \mathbf{a_3}$ are constants. Thus, the lower bound is an affine combination of these three constants, and the minimum is achieved in one of the following four cases:
\begin{align*}
    \frac{1}{\x_b-x_0+1}\sum_{\x=\x_0}^{\x_b}\sum_{\x_{u^*=1}}^{\numpieces}\frac{1}{\numpieces}\hclassbackup(\x_u,\x,\x_b,\x_{u^*})&&\text{(when $\mathbf{a_1},\mathbf{a_3}$ are maximized)},\\
    \frac{1}{\x_b-x_0}\sum_{\x=\x_0}^{\x_b-1}\sum_{\x_{u^*=1}}^{\numpieces}\frac{1}{\numpieces}\hclassbackup(\x_u,\x,\x_b,\x_{u^*})&&\text{(when $\mathbf{a_1}$ maximized ,$\mathbf{a_3}$ minimized)},\\
    \frac{1}{\x_b-x_0}\sum_{\x=\x_0+1}^{\x_b}\sum_{\x_{u^*=1}}^{\numpieces}\frac{1}{\numpieces}\hclassbackup(\x_u,\x,\x_b,\x_{u^*})&&\text{(when $\mathbf{a_1}$ minimized ,$\mathbf{a_3}$ maximized)},\\
    \frac{1}{\x_b-x_0-1}\sum_{\x=\x_0+1}^{\x_b-1}\sum_{\x_{u^*=1}}^{\numpieces}\frac{1}{\numpieces}\hclassbackup(\x_u,\x,\x_b,\x_{u^*})&&\text{(when $\mathbf{a_1}, \mathbf{a_3}$ are minimized)}.
\end{align*}
Now we can derive a lower bound for the expected gain conditioning on $\sigma\in EC_s^i$. We assume $\x_u$ is fixed and $\x_0,\x_b$ are chosen to be adversarial. The lower bound is the following, accounting for all the four different cases above:
$$
    \min_{1\leq c\leq d\leq k}\left\{\frac{1}{d-c+1}\sum_{\x=c}^{d}\sum_{\x_{u^*=1}}^{\numpieces}\frac{1}{\numpieces}\hclassbackup(\x_u,\x,d+1,\x_{u^*})\right\}.
$$
$d+1$ instead of $d$ appears here, because in cases where $\mathbf{a_3}$ is minimized, the outer sum stops at $\x_b - 1$, while the inner parameter is $\x_b$. By the way we defined $\hclassbackup$ function as in \cref{lem:hclassbackup}, increasing $\x_b$ only decreases the function value. Thus, we can fit all four cases into the form where the outer sum sums up to $d$, and the inner parameter takes $d+1$.\footnote{\label{dplusone}This adversarial form introduces problem when $d=k$, as we use $k+1$ as a parameter in $\hclassbackup$. In our final $LP$, we will take the price function $\price$ and the function $\hclassbackup$ to be defined over domain $\{1,...,k+1\}$. As long as the monotonicity properties are enforced, extending the domain by one is not a problem. One can also avoid dealing with this by setting up an edge case lower bound so specifically insist $\classbackup$ takes in $d$ as parameter when $d=k$. We are not doing it here for aesthetic purposes and the fact that such a change does not affect the final LP solution. }

Similar to the $EC_s^i$ case, this is technically a lower bound stronger than needed, as we are choosing $\x_0$ and $\x_b$ to be adversarial. We are doing this because in the next step, we need to consider adversarial $\x_0,\x_b$ values but we are safe to assume a fixed $\x_u$ value.
\end{proof}
Any $\sigma\in \classbackup$ belongs to some equivalence class of type $EC_b^i$, a lower bound for expected gain conditioned on $\sigma\in\classbackup$ and $\x_u=i$ can be obtained by setting $\x_u=i$ as the following:
\begin{claim}\label{claim7.3}
     Let $\sigma$ be the bucketed random rank vector over $V\backslash\{u^*\}$, and let $i$ be a bucket value in $\{1,...,k\}$. We have that the expected value of $g(u)+g(u^*)$ conditioning on $\sigma\in \classbackup$ and $\x_u=i$, i.e. the value
     $$\ev[g(u)+g(u^*)\;|\; \sigma\in \classbackup\;\wedge\;\x_u=i],$$
     is lower bounded by:
     \begin{equation}\label{expectedhb}
     \min_{1\leq c\leq d\leq k}\left\{\frac{1}{d-c+1}\sum_{\x=c}^{d}\sum_{\x_{u^*=1}}^{\numpieces}\frac{1}{\numpieces}\hclassbackup(i,\x,d+1,\x_{u^*})\right\}.
     \end{equation}
 \end{claim}
 \begin{proof}
 The proof is similar to the case $\classnomatch$. See detailed proof in \nameref{Appendix1}.
 \end{proof}
\subsection{Combining the Lower Bounds}
We will combine the three lower bounds defined in this section to obtain a lower bound for the approximation ratio for \ranking{} on general graphs.
\begin{claim}\label{alphaislowerbound}
    Suppose for each $i\in \{1,...,k\}$, $\alpha_i\in[0,1]$ satisfies the following three inequalities:
\begin{align*}
    \alpha_i\leq&\sum_{\x_{_{u^*}}=1}^{\numpieces}\frac{1}{k}\hclassnomatch(i,\x_{u^*}),&\cref{expextedhbot}\\
    \alpha_i\leq& \min_{c\in\{1,...,\numpieces\}}\left\{\frac{1}{\numpieces-c+1}\sum_{\x=c}^{\numpieces}\sum_{\x_{u^*=1}}^{\numpieces}\frac{1}{\numpieces}\hclassnobackup(i,x,\x_{u^*})\right\},&\cref{expectedhs}\\
    \alpha_i\leq& \min_{1\leq c\leq d\leq k}\left\{\frac{1}{d-c+1}\sum_{\x=c}^{d}\sum_{\x_{u^*=1}}^{\numpieces}\frac{1}{\numpieces}\hclassbackup(i,x,d+1,\x_{u^*})\right\}.&\cref{expectedhb}
\end{align*}
Then $\alpha_i$ is a lower bound for the approximation ratio for $\ranking{}$ conditioning on $\x_u=i$. Let $\alpha$ be the expected value of $\alpha_i$. That is
\begin{equation}\label{expectedalpha}
    \alpha=\sum_{i=1}^{\numpieces}\frac{1}{k}\alpha_i.
\end{equation}
Then $\alpha$ is a lower bound for the approximation ratio for \ranking{} on general graphs.
\end{claim}
\begin{proof}
By \cref{claim7.1}, \cref{claim7.2}, and \cref{claim7.3}, as $\alpha_i$ satisfies \cref{expextedhbot}, \cref{expectedhs}, \cref{expectedhb}, it holds that
\begin{align*}
    \alpha_i\;\leq\; &\ev[g(u)+g(u^*)\;|\; \sigma\in \classnomatch\;\wedge\;\x_u=i]\cdot\pr(\sigma\in \classnomatch\;|\;\x_u=i)\\
    &+\ev[g(u)+g(u^*)\;|\; \sigma\in \classnobackup\;\wedge\;\x_u=i]\cdot\pr(\sigma\in \classnobackup\;|\;\x_u=i)\\
    &+\ev[g(u)+g(u^*)\;|\; \sigma\in \classbackup\;\wedge\;\x_u=i]\cdot\pr(\sigma\in \classbackup\;|\;\x_u=i),\\
    \leq \;&\ev[g(u)+g(u^*)\;|\; \x_u=i] \quad\quad\quad\quad\text{(By law of total Expectation)}
\end{align*}
 Further, because $\x_u$ is a bucket value chosen uniformly from \{1,...,k\}, $\alpha$ satisfies
\begin{align*}
    \alpha=\sum_{i=1}^\numpieces \frac{1}{k}\alpha_i\leq\sum_{i=1}^\numpieces \ev[g(u)+g(u^*)\;|\; \x_u=i]\cdot \pr(\x_u=i)=\ev[g(u)+g(u^*)].
\end{align*}
By \cref{guustarisalowerbound}, such $\alpha$ is a valid lower bound for \ranking{} on general graphs.
\end{proof}

\section{The LP}\label{LP}
In this section, we present the full LP that integrates all the bounds we have set. Each feasible solution has its output objective $\alpha$ serving as a valid lower bound. The LP searches through all valid price functions $\price:\{1,\dots,\numpieces+1\}^2 \to [0,1]$ and calculates the corresponding $\alpha$ value as output. We used $k+1$ instead of $k$\footnotemark
\footnotetext{See \cref{dplusone} for reference.} to account for the $d+1$ appearing in \cref{expectedhb}. The following $\forall x$, without specifically mentioning, refers to $\forall x$ s.t. $1\leq x\leq k+1$.
\begin{align*}
    &&&\max_{f:\{1,\dots,\numpieces+1\}^2\to [0,1]} \alpha\\
    \text{s.t.} \quad 
    &\forall  i,\forall j\leq \numpieces  & f(i,j) &\leq f(i,j+1),&\cref{pricefunction}\\
    &\forall i\leq\numpieces,\forall j  & f(i,j) &\geq f(i+1,j),\\
    &\forall u,u^* & \hclassnomatch(\x_u,\x_{u^*})&=\price(\x_u,\x_{u^*}).&\cref{defhclassnomatch}\\
    &\forall u,v,u^*  & \hclassnobackup(\x_u,\x_v,\x_{u^*}) &= 
    \begin{cases} 
        \price(\x_u,\x_{u^*}), & \x_u < \x_v, \x_{u^*} \leq \x_u, \\
        1 - \price(\x_u,\x_v) + \price(\x_u,\x_{u^*}), & \x_u < \x_{u^*} < \x_v, \\
        1 - \price(\x_u,\x_v), & \x_u < \x_v, \x_v \leq \x_{u^*}, \\
        \min\{\price(\x_v,\x_{u^*}), & \x_v \leq \x_u, \x_{u^*} < \x_v, \\
        \quad \quad1 - \price(\x_u,\x_v) +\price(\x_u,\x_{u^*})\}, \\
        \min\{\price(\x_v,\x_{u^*}), 1 - \price(\x_u,\x_v)\},  & \x_v \leq \x_{u^*} \leq \x_u,\\
        1 - \price(\x_u,\x_v), & \x_v \leq \x_u , \x_u < \x_{u^*}.
    \end{cases}&\cref{defhclassnobackup}\\
     &\forall \x_u, \x_v,\x_b,\x_{u^*}  & \hclassbackup(\x_u,\x_v,\x_b,\x_{u^*}) &= 
    \begin{cases} 
        1 - \price(\x_u,\x_b) + \price(\x_u,\x_{u^*}), & \x_u < \x_v, \x_{u^*} \leq \x_u, \\
        1 - \price(\x_u,\x_v) + \price(\x_u,\x_{u^*}), & \x_u < \x_{u^*} < \x_v, \\
        1 - \price(\x_u,\x_v), & \x_u < \x_v, \x_v \leq \x_{u^*}, \\
        \min\{1 - \price(\x_u,\x_b) + \price(\x_v,\x_{u^*}), \\
        \quad\quad1 - \price(\x_u,\x_v) + \price(\x_u,\x_{u^*})\}, & \x_v \leq \x_u, \x_{u^*} < \x_v, \\
        \min\{1 - \price(\x_u,\x_b) + \price(\x_v,\x_{u^*}), \\
        \quad\quad1 - \price(\x_u,\x_v)\}, & \x_v \leq \x_{u^*} \leq \x_u, \\
        1 - \price(\x_u,\x_v), & \x_v \leq \x_u , \x_u < \x_{u^*}.
    \end{cases}&\cref{defhclassbackup}\\
    &\forall i,  & \alpha_i &\leq \sum_{\x_{u^*}=1}^{\numpieces} \frac{1}{\numpieces} \hclassnomatch(i,\x_{u^*}),&\cref{expextedhbot}\\
    &\forall i,\; \forall c \leq \numpieces  & \alpha_i &\leq  \frac{1}{\numpieces-c+1} \sum_{\x_v=c}^\numpieces\sum_{\x_{u^*}=1}^{\numpieces} \frac{1}{\numpieces} \hclassnobackup(i,\x_v,\x_{u^*}),&\cref{expectedhs}\\
    &\forall i,\; \forall c\leq d\leq\numpieces, & \alpha_i &\leq  \frac{1}{d-c+1} \sum_{\x_v=c}^{d}\sum_{\x_{u^*}=1}^{\numpieces} \frac{1}{\numpieces} \hclassbackup(i,\x_v,d+1,\x_{u^*}),&\cref{expectedhb}\\
    && \alpha &= \sum_{i=1}^{\numpieces} \frac{1}{\numpieces}\alpha_i,&\cref{expectedalpha}\\
\end{align*}
Each constraint in the LP corresponds to:

\cref{pricefunction}: This constraint corresponds to the monotonicity constraints for the price function defined in  \cref{pricefun}.

\cref{defhclassnomatch}: This constraint corresponds to the function $\hclassnomatch$ that provides a lower bound for $g(u) + g(u^*)$ conditioning on a rank vector $\sigma_{-u^*}$ with profile for $u=(\x_u,\bot,\bot)$ in $\classnomatch$ being realized. This bound is defined and proved in \cref{sectionhbot}.

\cref{defhclassnobackup}:This constraint corresponds to the function $\hclassnobackup$ that provides a lower bound for $g(u) + g(u^*)$ conditioning on a rank vector $\sigma_{-u^*}$ with profile for $u=(\x_u,\x_v,\bot)$ in $\classnobackup$ being realized. The terms of the form $\hclassnobackup=\min\{a_1,a_2\}$ that appeared in the constraint can be achieved by setting two bounds $\hclassnobackup\leq a_1\;\wedge \hclassnobackup\leq a_1$. This bound is defined and proved in \cref{sectionhs}.

\cref{defhclassnobackup}:This constraint corresponds to the function $\hclassbackup$ that provides a lower bound for $g(u) + g(u^*)$ conditioning on a rank vector $\sigma_{-u^*}$ with profile for $u=(\x_u,\x_v,\x_b)$ in $\classbackup$ being realized. The terms of the form $\hclassbackup=\min\{a_1,a_2\}$ appeared in the constraint can be achieved by setting two bounds $\hclassbackup\leq a_1\;\wedge \hclassbackup\leq a_1$. This bound is defined and proved in \cref{sectionhb}.

\cref{expextedhbot}: This constraint corresponds to the lower bound for expected gain $g(u)+g(u^*)$ conditioning on $\sigma_{-u^*}\in \classnomatch$ and $\x_u=i$. This bound is defined and proved in \cref{sectionexpecthbot}.

\cref{expectedhs}: This constraint corresponds to the lower bound for expected gain $g(u)+g(u^*)$ conditioning on $\sigma_{-u^*}\in \classnobackup$ and $\x_u=i$. This bound is defined and proved in \cref{sectionexpecths}.

\cref{expectedhb}: This constraint corresponds to the lower bound for expected gain $g(u)+g(u^*)$ conditioning on $\sigma_{-u^*}\in \classbackup$ and $\x_u=i$. This bound is defined and proved in \cref{sectionexpecthb}.

\cref{expectedalpha}: This constraint sets $\alpha$ as the average of $\alpha_i$, which by \cref{alphaislowerbound}, combined with all other constraints, forces $\alpha$ to be a lower bound for the approximation ratio for \ranking{} on general graphs.

\section{Numerical Solution}\label{Numericalapproximation}

In this section, we present the results of running the LP for $k \leq 100$. All code is written in Python (version 3.10.12) and is available upon request. To solve the LP instances, we used the Gurobi optimization package (version 11.0.0). The experiments were conducted on a computing cluster featuring 64 cores, each running at 2.30 GHz on Intel(R) Xeon(R) processors, with 756 GiB of main memory. The operating system was Ubuntu 22.04.3 LTS. Solving the LP for $k = 100$ takes approximately 15 hours. The results of the LP implementation are shown in \Cref{tab:LP-solution}.  Also, the values of function $f$ is shown in \Cref{tab:f_values} for $k=10$.

\vspace{2em}
\begin{table}[h]
  \centering
  \begin{tabular}{|c|c||c|c||c|c||c|c|}
    \hline
    $k$ & LP Solution & $k$ & LP Solution & $k$ & LP Solution & $k$ & LP Solution\\ 
    \hline
    1  & 0.5  &  11 & 0.53202  & 25 & 0.54098 & 75 & 0.54624\\ \hline
    2  & 0.5  &  12 & 0.53334  & 30 & 0.54226 & 80 & 0.54641\\ \hline
    3  & 0.50347  &  13 & 0.53443  & 35 & 0.54318 & 85 & 0.54656\\ \hline
    4  & 0.51052  &  14 & 0.53530  & 40 & 0.54389 & 90 & 0.54669\\ \hline
    5  & 0.51625  &  15 & 0.53608  & 45 & 0.54444 & 95 & 0.54681\\ \hline
    6  & 0.52068  &  16 & 0.53687  & 50 & 0.54489 & 100 & 0.54690 \\ \hline
    7  & 0.52422  &  17 & 0.53755  & 55 & 0.54525 & - & -\\ \hline
    8  & 0.52674  &  18 & 0.53812  & 60 & 0.54556 & - & -\\ \hline
    9  & 0.52882  &  19 & 0.53863  & 65 & 0.54582 & - & -\\ \hline
    10  & 0.53046  &  20 & 0.53910  & 70 & 0.54604 & - & -\\ \hline
  \end{tabular}
  \caption{The LP solution for various values of $k$.}
  \label{tab:LP-solution}
\end{table}

\begin{table}[h]
    \centering
    \renewcommand{\arraystretch}{1.2} 
    \setlength{\tabcolsep}{6pt} 
    \small
    \begin{tabular}{|c||c|c|c|c|c|c|c|c|c|c|c|}
        \hline
        $i \backslash j$ & 1 & 2 & 3 & 4 & 5 & 6 & 7 & 8 & 9 & 10 & 11 \\ \hline\hline
        1 & 0.423 & 0.462 & 0.500 & 0.553 & 0.585 & 0.629 & 0.632 & 0.797 & 0.843 & 0.843 & 1.000 \\ \hline
        2 & 0.423 & 0.462 & 0.500 & 0.538 & 0.585 & 0.629 & 0.632 & 0.632 & 0.684 & 0.684 & 1.000 \\ \hline
        3 & 0.423 & 0.462 & 0.500 & 0.538 & 0.585 & 0.629 & 0.632 & 0.632 & 0.632 & 0.632 & 1.000 \\ \hline
        4 & 0.423 & 0.462 & 0.500 & 0.538 & 0.585 & 0.599 & 0.599 & 0.599 & 0.599 & 0.599 & 1.000 \\ \hline
        5 & 0.423 & 0.462 & 0.500 & 0.538 & 0.570 & 0.570 & 0.570 & 0.570 & 0.570 & 0.570 & 1.000 \\ \hline
        6 & 0.423 & 0.462 & 0.500 & 0.538 & 0.544 & 0.544 & 0.544 & 0.544 & 0.544 & 0.544 & 1.000 \\ \hline
        7 & 0.423 & 0.462 & 0.500 & 0.523 & 0.523 & 0.523 & 0.523 & 0.523 & 0.523 & 0.523 & 1.000 \\ \hline
        8 & 0.423 & 0.462 & 0.500 & 0.501 & 0.512 & 0.512 & 0.512 & 0.512 & 0.512 & 0.512 & 1.000 \\ \hline
        9 & 0.423 & 0.462 & 0.500 & 0.501 & 0.505 & 0.506 & 0.510 & 0.510 & 0.510 & 0.510 & 1.000 \\ \hline
        10 & 0.423 & 0.462 & 0.500 & 0.501 & 0.505 & 0.505 & 0.505 & 0.505 & 0.505 & 0.505 & 1.000 \\ \hline
        11 & 0.000 & 0.000 & 0.000 & 0.000 & 0.000 & 0.000 & 0.000 & 0.000 & 0.000 & 0.000 & 0.000 \\ \hline
    \end{tabular}
    \caption{Values of $f(i,j)$ for $k=10$.}
    \label{tab:f_values}
\end{table}

\section{Main Theorem}
\begin{theorem}
    The approximation ratio for \ranking{} on general graphs is at least $0.5469$.
\end{theorem}
\begin{proof}
    By \cref{alphaislowerbound}, any feasible solution to the LP presented in \cref{LP}, with the parameter $k$ set to an arbitrary value, serves as a lower bound for the approximation ratio for \ranking{} on general graphs. We ran the LP with $k=100$ and obtained a feasible solution $\alpha=0.5469$. This implies that the approximation ratio for \ranking{} on general graphs is at least $0.5469$.
\end{proof}

\section*{Deferred Proofs}

\begin{proof}[\textbf{Proof for }\cref{same-match-before-v}]
For $\M^t(\sigma)= \M^t(\sigma^+)$: Assume towards a contradiction that $\M^t(\sigma)\neq \M^t(\sigma^+)$. Let $t_0<t$ be the earliest time the two sets become different. Then $\M^{t_0}(\sigma)= \M^{t_0}(\sigma^+)$ and $A^{t_0}(\sigma)= A^{t_0}(\sigma^+)$ but not in the successive time. At $t_0$, the vertex $u_0$ being processed picked some vertex in one matching, the same vertex is also available in the other permutation, but $u_0$ will not pick the same vertex as the two matchings becomes different at the next time stamp. Let $(u_0,u_1)\in \M^t(\sigma)$ and $(u_0,u_2)\in \M^t(\sigma^+)$ be the different edges picked by $u_0$. In this case, $u_1$ is preferred over $u_2$ in $\sigma$ but not in $\sigma^+$. The only possible choices for $u_1$ and $u_2$ are $v$ and $v'$. This is because the only vertex that becomes relatively less preferred is $v$, and it only becomes less preferred compared to $v'$ but not any other nodes. However, by assumption, $v$ cannot be matched to $u_0$ at time $t_0<t$ as $v$ is still available in $A^t(\sigma)$, leading to a contradiction.

For $\M^t(\sigma)=\M^t(\sigma_{-v})$: Since $t\in A^t(\sigma)$, it is never the best choice for any vertex with rank $<t$, so whether $v$ is available or not does not affect any choice made by these vertices. This implies $\M^t=\M^t_{\sigma_{-v}}$.

    Now that we have $\M^t=\M^t_+=\M^t_{\sigma_{-v}}$, it is easy to see that $A^t=A^t_{+}=A^t_{-v}\cup\{v\}$ by the definition of available vertices at time $t$.
\end{proof}

\begin{proof}[\textbf{Proof for }\cref{descrete-monotonicity-backup}]\label{Appendix3}
Let $t$ be the time right before $v$ becomes unavailable in the \ranking{} process on $\sigma$. It's easy to see $t$ is the smaller of $\sigma(u)$ and $\sigma(v)$. We consider three different cases:

Case 1: $t=\sigma(u) < \sigma(v)$. Then $t=\sigma(u)$, which is the time when $u$ picks $v$. Since $v$ is still available at this time, \cref{same-match-before-v} implies $\M^t(\sigma) = \M^t(\sigma^+)$ and ${u, v} \subseteq A^t(\sigma) = A^t(\sigma^+)$. It suffices to show that $v$ is still the best available neighbor of $u$ at time $t$. By \cref{same-match-before-v}, since $A^t(\sigma) = A^t(\sigma_{-v})\cup\{v\}$ the second best available neighbor at this time is $b$, with $\sigma(b)>\sigma(v)+1$. As $\sigma^+(v)=\sigma(v)+1<\sigma(b)=\sigma^+(b)$, $v$ is still the best available neighbor after the increase.

Case 2:  $t=\sigma(v) = \sigma(u) - 1$. Here, $\sigma^+$ simply swaps the order of $u$ and $v$. By \cref{same-match-before-v}, $u$ and $v$ are both available in $A^t(\sigma^+)$. At time $t$, $\sigma^+(u) = t$ will actively pick its smallest ranked neighbor according to $\sigma^+$. Since $\sigma^+(v) = \sigma^+(u) + 1$, $v$ is necessarily the smallest ranked vertex among the available neighbors of $u$, and will be picked by $u$.

Case 3: $t=\sigma(v) < \sigma(u) - 1$. By \cref{same-match-before-v}, we have $A^t(\sigma) = A^t(\sigma^+)$. We will show that the vertex being processed at time $t$ by \ranking{} on $\sigma^+$, namely $v'$, will not pick $u$ or $v$ as its match. If this is the case, then $v$ will being the next vertex visited will be alive and has $u$ as an available neighbor. Since the relative ordering of vertices with rank $> \sigma(v)+1$ is not changed between $\sigma$ and $\sigma^+$, and all the available neighbor of $v$ at time $t+1$ fall in this range, $v$ will still pick $u$ as its match. 

Suppose $v'$ is not matched yet so that it might pick $u$ or $v$ to match. First, we claim that $v'$ is not connected to $v$. If $v$ and $v'$ were connected, then $v' \in A^t(\sigma^+) = A^t(\sigma)$ would be the best choice for $v$ in $\sigma$ when $v$ picks. This contradicts the assumption that $v$ picks $u$ as its match in $\M(\sigma)$.

On the other hand, assume $v'$ picks $u$. Then $u$ must be the best neighbor for $v'$ in $A^t(\sigma^+)$. By \cref{same-match-before-v}, the set of available neighbors in $A^t(\sigma_{-v})$ is also the same. Furthermore, the relative ordering of this set of available neighbors remains unchanged, so $v'$ would also pick $u$ in $\M(\sigma_{-v})$. Hence $v'$ is the backup of $u$ with respect to the ordering $\sigma$. But by assumption $\sigma(b)>\sigma(v)+1=\sigma(v')$. This is a contradiction.

Finally, we need to show that $u$ still have the backup $b$ in $\sigma^+$. This follows from the fact that $\sigma_{-v}$ and $(\sigma^+)_{-v}$ are the same ordering on the same induced graph, meaning \ranking{} will generate the same matching.
\end{proof}

\begin{proof}[\textbf{Proof for }\cref{bucketedpermutation}] \label{Appendix1}
    Assume $|V| = n$. Let $\sigma_0 : V \to \{1, \dots, n\}$ be any permutation. Let $\vecx, \vecy$ be the random vector of $\x_v, \y_v$ values generated by the bucketed random permutation. 
    We will prove by induction on the number of buckets $\numpieces$ that 
    $$ \pr(\sigma_{\vecx,\vecy} = \sigma_0) = \frac{1}{n!}. $$

    When the number of buckets $k = 1$, the result holds trivially because all nodes fall in the same bucket, and we permute all the vertices uniformly at random. 

    Assume it holds for $k - 1$ buckets. The probability of realizing $\sigma$ with $k$ buckets is
    \begin{align*}
        \pr(\sigma_{\vecx,\vecy}=\sigma)&=\sum_{i=0}^n\pr(\text{first $i$ nodes in $\sigma$ falls in the first bucket}\wedge \sigma_{\vecx,\vecy}=\sigma_0),\\
        &=\sum_{i=0}^n (\frac{1}{\numpieces})^i(\frac{\numpieces-1}{\numpieces})^{n-i}\frac{1}{i!}\frac{1}{(n-i)!},\\
        &=\frac{1}{n!}\sum_{i=0}^n \frac{n!}{i!(n-i)!} (\frac{1}{\numpieces})^i(\frac{\numpieces-1}{\numpieces})^{n-i},\\
        &=\frac{1}{n!}(\frac{1}{\numpieces}+\frac{\numpieces-1}{\numpieces})^n,\\
        &=\frac{1}{n!}.
    \end{align*}
    The second equality follows from the fact that
    $$\pr(\sigma_{\vecx,\vecy}=\sigma_0\;|\;\text{first $i$ nodes in $\sigma_0$ fall in the first bucket})=\frac{1}{i!}\frac{1}{(n-i)!},$$
    as the probability that the first $i$ vertices in bucket 1 realize the same ordering as the first $i$ vertices in $\sigma_0$ is $\frac{1}{i!}$. The probability that the remaining vertices falling into buckets $2, \dots, k$ realize the same ordering as the remaining vertices in $\sigma$ is $\frac{1}{(n-i)!}$ by the inductive hypothesis.

The second-to-last inequality follows from the binomial formula.
\end{proof}

\begin{proof}[\textbf{Proof for }\cref{sameexceptv}]
   $EC_\bot^i$: This is trivial.

   $EC_s^i$: Fix arbitrary $\sigma\in EC_s^i$, let $v$ be the match of $u$ in $\M(\sigma)$. There exists minimum $(\x_0,\y_0)$ s.t. $\sigma_v^{(\x_0,\y_0)}\in EC_s^i$. Denote $\sigma_0=\sigma_v^{(\x_0,\y_0)}$. By the definition of $\sim$, we have that the matches of $u$ in $\M(\sigma_0)$ and $\M(\sigma)$ both equal to $v$. By \cref{descrete-monotonicity-nobackup}, moving $v$ to any $(x,y)\geq (\x_0,\y_0)$ will not change the fact that $u$ is matched to $v$. So $v$ is also the match of $u$ in such $(\sigma_0)_v^{(x,y)}$ and hence we also have $(\sigma_0)_{-v}=((\sigma_0)^{(x,y)}_v)_{-v}$. This $\sigma_0$ is the generating vector as defined.

   $EC_b^i$: Fix arbitrary $\sigma\in EC_b^i$. Let $v$ be the match of $u$ in $\M(\sigma)$. There exists minimum $(\x_0,\y_0)$ s.t. $\sigma_v^{(\x_0,\y_0)}\in EC_s^i$. Denote $\sigma_0=\sigma_v^{(\x_0,\y_0)}$. By the definition of $\sim$, we have that the matches of $u$ in $\M(\sigma_0)$ and $\M(\sigma)$ both equal to $v$. Denote $b$ as the backup of $u$ in $\sigma_0$. By \cref{descrete-monotonicity-backup}, permuting $v$ to any $(x,y)$ s.t. $(\x_0,\y_0)\leq(x,y)<\sigma_0(b)$ will not change the fact that $u$ is matched to $v$ and the fact that $b$ is the backup of $u$. So $v$ is also the match of $u$ in such $\sigma_v^{(x,y)}$ and hence we also have $(\sigma_0)_{-v}=((\sigma_0)^{(x,y)}_v)_{-v}$. This $\sigma_0$ is the generating vector as defined.
\end{proof}

\begin{proof}[\textbf{Proof for }\cref{claim7.2}]
    By the definition of $\sim$ and \cref{sameexceptv}, there exists a list of equivalence classes $EC_s^{i_1}, EC_s^{i_2}, \ldots, EC_s^{i_l}$ partitioning the set of rank vectors $\{\sigma \;|\; \sigma \in \classnobackup \;\wedge\; \x_u = i\}$. Let 
$$\alpha_i = \min_{c\in\{1,...,\numpieces\}}\left\{\frac{1}{\numpieces-c+1}\sum_{\x=c}^{\numpieces}\sum_{\x_{u^*=1}}^{\numpieces}\frac{1}{\numpieces}\hclassnobackup(i,\x,\x_{u^*})\right\}.$$
We have:
$$
\alpha_i \leq \ev[g(u) + g(u^*) \;|\; \sigma \in EC_s^{i_j}] \quad \text{for all equivalence classes } EC_s^{i_j}.
$$
Thus:
\begin{align*}
    \alpha_i &\leq \sum_{j=1}^l \ev[g(u) + g(u^*) \;|\; \sigma \in EC_s^{i_j}] \cdot \pr(\sigma \in EC_s^{i_j} \;|\; \sigma \in \classnobackup \;\wedge\; \x_u = i), \\
    &= \ev[g(u) + g(u^*) \;|\; \sigma \in \classnobackup \;\wedge\; \x_u = i],
\end{align*}
by the law of total expectation.
\end{proof}

\begin{proof}[\textbf{Proof for }\cref{claim7.3}]
     By the definition of $\sim$ and \cref{sameexceptv}, there exists a list of equivalence classes $EC_b^{i_1}, EC_b^{i_2}, \ldots, EC_b^{i_l}$ partitioning the set of rank vectors $\{\sigma \;|\; \sigma \in \classbackup \;\wedge\; \x_u = i\}$. Let 
     $$\alpha_i = \min_{c\leq d}\left\{\frac{1}{d-c+1}\sum_{\x=c}^{d}\sum_{\x_{u^*=1}}^{\numpieces}\frac{1}{\numpieces}\hclassbackup(i,\x,d+1,\x_{u^*})\right\}.$$ 
     We have:
$$
\alpha_i \leq \ev[g(u) + g(u^*) \;|\; \sigma \in EC_b^{i_j}] \quad \text{for all equivalence classes } EC_b^{i_j}.
$$
Thus:
\begin{align*}
    \alpha_i &\leq \sum_{j=1}^d \ev[g(u) + g(u^*) \;|\; \sigma \in EC_b^{i_j}] \cdot \pr(\sigma \in EC_b^{i_j} \;|\; \sigma \in \classbackup \;\wedge\; \x_u = i), \\
    &= \ev[g(u) + g(u^*) \;|\; \sigma \in \classbackup \;\wedge\; \x_u = i],
\end{align*}
by the law of total expectation.
\end{proof}

\section*{Appendix}
\subsection{Comparison of Backup Vertices with Prior Work}\label{sec:appendix-a}
To avoid leaving \(u\) unmatched when \(u^*\) is added, previous work \cite{chan2018analyzing} considered an idea similar to backups. In their Lemma~3.6, they stated the following fact:
\begin{lemma}
    Let $\sigma$ be an ordering $\sigma:V\to \{1,...,|V|\}$ Suppose $u$ is matched and has two unmatched neighbors $w_1,w_2$ in $\M(\sigma)$, then $u$ is matched even if we demote the rank of $u$ to the bottom.
\end{lemma}
This lemma also guarantees, to some extent, that \(u\) will always be matched. With a bit more effort, one can infer from Lemma~3.6 that if \(u\) has an unmatched neighbor \(w\) in \(\M(\sigma_{-u^*})\), then when \(u^*\) is added, \(u\) will always match to some vertex at least as good as \(w\). Intuitively, this suggests defining the backup \(b\) as the first unmatched neighbor of \(u\) in \(\M(\sigma_{-u^*})\).

It is true that if unmatched neighbors exist, then \(u\) will have a backup. However, it is not true that all backups are unmatched neighbors. Consider the following partial graph $G_{\text{partial}}$ with ordering \(\sigma_{-u^*}\), where a larger rank corresponds to a position lower in the ordering:
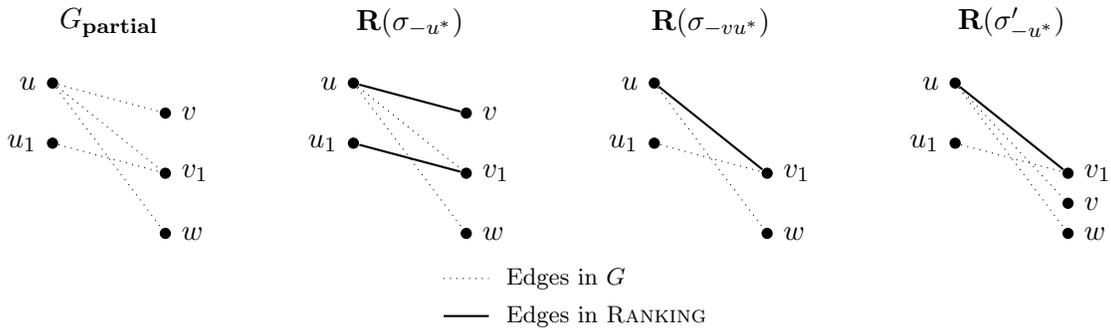
\newline
\begin{center}
\begin{tikzpicture}[dot/.style={circle, fill, inner sep=1.5pt}]

\node at (0.75,2)  {\textbf{$G_{\text{partial}}$}};
\node at (4.75,2)  {\textbf{$\M(\sigma_{-u^*})$}};
\node at (8.75,2)  {\textbf{$\M(\sigma_{-vu^*})$}};
\node at (12.75,2) {\textbf{$\M(\sigma_{-u^*}')$}};  

\node[dot, label=left:$u$]   (u1) at (0,1.2) {};
\node[dot, label=left:$u_1$] (u2) at (0,0.4) {};
\node[dot, label=right:$v$]   (v)  at (1.5,0.8) {};
\node[dot, label=right:$v_1$] (v1) at (1.5,0) {};
\node[dot, label=right:$w$]   (w)  at (1.5,-0.8) {};
\draw[dotted] (u1) -- (v);
\draw[dotted] (u1) -- (v1);
\draw[dotted] (u1) -- (w);
\draw[dotted] (u2) -- (v1);

\begin{scope}[xshift=4cm]
\node[dot, label=left:$u$]   (u1) at (0,1.2) {};
\node[dot, label=left:$u_1$] (u2) at (0,0.4) {};
\node[dot, label=right:$v$]   (v)  at (1.5,0.8) {};
\node[dot, label=right:$v_1$] (v1) at (1.5,0) {};
\node[dot, label=right:$w$]   (w)  at (1.5,-0.8) {};
\draw[thick]  (u1) -- (v);
\draw[dotted] (u1) -- (v1);
\draw[dotted] (u1) -- (w);
\draw[thick]  (u2) -- (v1);
\end{scope}

\begin{scope}[xshift=8cm]
\node[dot, label=left:$u$]   (u1) at (0,1.2) {};
\node[dot, label=left:$u_1$] (u2) at (0,0.4) {};
\node[dot, label=right:$v_1$] (v1) at (1.5,0) {};
\node[dot, label=right:$w$]   (w)  at (1.5,-0.8) {};
\draw[thick]  (u1) -- (v1);
\draw[dotted] (u1) -- (w);
\draw[dotted] (u2) -- (v1);
\end{scope}

\begin{scope}[xshift=12cm]
\node[dot, label=left:$u$]   (u1) at (0,1.2) {};
\node[dot, label=left:$u_1$] (u2) at (0,0.4) {};
\node[dot, label=right:$v_1$] (v1) at (1.5,0) {};
\node[dot, label=right:$w$]   (w)  at (1.5,-0.8) {};
\node[dot, label=right:$v$]   (v)  at (1.5,-0.4) {};
\draw[thick]  (u1) -- (v1);
\draw[dotted] (u1) -- (w);
\draw[dotted] (u1) -- (v);
\draw[dotted] (u2) -- (v1);
\end{scope}

\begin{scope}[xshift=5.5cm, yshift=-2cm]
\draw[dotted] (-0.3,0.6) -- +(0.6,0);
\node[anchor=west] at (0.4,0.6) {\footnotesize Edges in $G$};

\draw[thick] (-0.3,0.1) -- +(0.6,0);
\node[anchor=west] at (0.4,0.1) {\footnotesize Edges in \ranking{}};
\end{scope}
\end{tikzpicture}
\end{center}
We see that an unmatched neighbor \(w\) may not be the backup of \(u\) in \(\M(\sigma_{-u^*})\), and that the backup \(v_1\) in this case may not itself be an unmatched vertex in \(\M(\sigma_{-u^*})\). This situation essentially shows that \Cref{descrete-monotonicity-nobackup} and \Cref{descrete-monotonicity-backup} fail to hold when backups are defined as the smallest unmatched neighbor. In the example, as in graph $\M(\sigma_{-u^*}')$, if we demote \(v\) to a rank larger than \(\sigma_{-u^*}(v_1)\), then \(v\) is no longer the match of \(u\), violating \Cref{descrete-monotonicity-backup}. The monotonicity property is one of the key components that enables us to achieve the \(0.5469\) approximation. Therefore, a more delicate definition of backup is required in our work compared to the notion of unmatched neighbors.

\subsection{Restricted and Unrestricted Vertex Arrival Models}\label{vertexarrivalmodels}
\paragraph{Unrestricted Vertex Arrival (Oblivious Matching) Model.} In the study of \ranking{}, we may refer to two different kinds of vertex arrival models when not explicitly stated. \cref{alg:vertex-ranking-general} is one of them. This model assumes that each vertex $u$ arrives at time $\sigma(u)$ and is matched to the smallest-ranked unmatched neighbor according to $\sigma$, regardless of whether the neighbor has arrived or not. We call this the unrestricted (uniform random) vertex arrival model.

\paragraph{Restricted Vertex Arrival Model.} There also exists a more restricted version of the random vertex arrival model, where, when a vertex arrives, it only looks at the neighbors that have already arrived, and it gets matched to the smallest-ranked one that's still available. We call this the restricted (uniform random) vertex arrival model.
\newline\newline
Intuitively, this restriction should not affect the matching choice. We will prove a simple fact here to establish this.
\begin{fact}
    For the same ordering $\sigma:V\to\{1,...|V|\}$, restricted and unrestricted vertex arrival model of \ranking{} produce the same outputs.
\end{fact}
\begin{proof}
We adopt the greedy probing view for both models. The greedy probing view for the unrestricted model is stated in \cref{greeedyprobing}. The greedy probing view for the restricted model can be viewed as iterating over the following disjoint sets $S_1, S_2, \dots, S_{|V|}$ one by one and greedily picking the edges. $S_i$ is the ordered set
    $$S_i=\{(i,1),(i,2),...(i,i)\}.$$
We denote the two greedy probing orderings as $GP_{\text{unres}}$ and $GP_{\text{res}}$ for the unrestricted and restricted models, respectively. Let $M_{\text{res}}$ and $M_{\text{unres}}$ be the matching results respectively. Assume $M_{\text{res}}\neq M_{\text{unres}}$. We first make the following observation:

For any two edges $e_1 = (u, v_1)$ and $e_2 = (u, v_2)$ that share the same vertex $u$, the relative orders in which $GP_{\text{unres}}$ and $GP_{\text{res}}$ query them are the same. That is, one can show that both orderings query $e_1$ before $e_2$ if and only if $\sigma(v_1) < \sigma(v_2)$\footnote{The required proof for $GP_{\text{unres}}$ is given by \cref{lexicographicorder}, and one can apply a similar process to give a proof for $GP_{\text{res}}$ as well.}. 

Now let $e = (u, v)$ be the first edge\footnote{According to the ordering $GP_{\text{unres}}$, which is also the natural lexicographic ordering.} that witnesses $M_{\text{res}} \neq M_{\text{unres}}$, i.e. $(u, v)$ is the first edge that's included in one of the matching sets but not the other. This means that either $u$ or $v$ is not available when $(u,v)$ gets queried in the other set. W.L.O.G., assume $u$ is not available. This means that some other edge $(u, v')$ is queried before $(u, v)$ and got included in the other matching set. By the observation we made earlier, this implies $\sigma(v') < \sigma(v)$, and hence $(u, v')$ comes before $(u, v)$ in the lexicographic ordering. This contradicts the assumption that $(u, v)$ is the first edge witnessing $M_{\text{res}} \neq M_{\text{unres}}$.
\end{proof}

\bibliographystyle{plain}
\bibliography{references}
\end{document}